\documentclass[11pt]{article}

\usepackage{url}
\usepackage{bbm}
\usepackage{mathtools}
\usepackage{amssymb}
\usepackage{amsthm}
\usepackage{empheq}
\usepackage{latexsym}
\usepackage{eurosym}
\usepackage{dsfont}
\usepackage{appendix}
\usepackage{color} 
\usepackage[unicode]{hyperref}
\usepackage{frcursive}
\usepackage[utf8]{inputenc}
\usepackage[T1]{fontenc}
\usepackage{multirow}
\usepackage{todonotes}
\usepackage{lmodern}
\usepackage{anyfontsize}
\usepackage{pgfplots}
\usepackage{stmaryrd}
\usepackage{float}

\usepackage{graphicx}
\usepackage{caption}
\usepackage{subcaption}
\usepackage{enumerate}
\usepackage{breakcites}
\usepackage{booktabs}

\graphicspath{ {images/} }

\pgfplotsset{compat=newest}

\definecolor{red}{rgb}{0.7,0.15,0.15}
\definecolor{green}{rgb}{0,0.5,0}
\definecolor{blue}{rgb}{0,0,0.7}
\hypersetup{colorlinks, linkcolor={red},citecolor={green}, urlcolor={blue}}
            
\makeatletter \@addtoreset{equation}{section}

\newtheorem{theorem}{Theorem}

\newtheorem{theorem2}{Theorem}[section]

\newtheorem{proposition}[theorem2]{Proposition}

\newtheorem{remark}[theorem2]{Remark}
\newcommand{\comment}[1]{}
\usepackage{geometry}
\DeclareUnicodeCharacter{014D}{\=o}
\DeclareMathOperator*{\esssup}{ess\,sup}

\title{Clearing time randomization and transaction fees\\
for auction market design}
\author{Thibaut {\sc Mastrolia}\footnote{UC Berkeley, Department of Industrial Engineering and Operations Research, thibaut.mastrolia@berkeley.edu} ~and~ Tianrui {\sc Xu}\footnote{UC Berkeley, Department of Mathematics, tianrui.xu@berkeley.edu}}

\date{\today}

\usepackage{todonotes} 
\begin{document}

\maketitle
	\begin{abstract}
	Flaws of a continuous limit order book mechanism raise the question of whether a continuous trading session and a periodic auction session would bring better efficiency. This paper wants to go further in designing a periodic auction when both a continuous market and a periodic auction market are available to traders. In a periodic auction, we discover that a strategic trader could take advantage of the accumulated information available along the auction duration by arriving at the latest moment before the auction closes, increasing the price impact on the market. Such price impact moves the clearing price away from the efficient price and may disturb the efficiency of a periodic auction market. We thus propose and quantify the effect of two remedies to mitigate these flaws: randomizing the auction's closing time and optimally designing a transaction fees policy for both the strategic traders and other market participants. Our results show that these policies encourage a strategic trader to send their orders earlier to enhance the efficiency of the auction market, illustrated by data extracted from Alphabet and Apple stocks.  \\ \\
		
	\end{abstract}
\noindent\rule{15cm}{0.4pt}

\section{Introduction}

\subsection{Periodic auction and continuous limit order book}

Continuous limit order book (CLOB for short or continuous double auction) and periodic auction (also called batch auction or call auction) are the two most commonly used electronic trading systems around the world. For example, New York Stock Exchange, NASDAQ (U.S.), London Stock Exchange all use the continuous limit order book system during normal trading hours and switch to a periodic auction system for the opening and closing auctions to determine the open price and the closing price of each trading day. CBOE Europe have both continuous limit order book and periodic auction open during the normal trading hours. CowSwap a cryptocurrency exchange uses periodic auction to settle orders. %\vspace{0.5em} 

\bigskip

The trading mechanisms of the two systems are as follows. A continuous limit order book market executes incoming orders continuously, i.e., when there is a matching order between a market order and a limit order. Every matched order trades at a price that depends on its requested price and the price of the limit order it is matched to. In comparison, a periodic auction market executes incoming orders as a batch and applies a uniform price to all executed orders in this batch after a specific time horizon. To be more specific, a market order would initiate an auction that would be open for a specific time interval until a terminal time, which is called the clearing time. During the auction, the exchange receives orders from market participants. Market participants give the exchange a proposed price at which they are willing to buy or sell the asset and a specific volume. When the auction closes, the exchange determines a clearing price by setting it to maximize the number of fulfilled orders (or to minimize the imbalance). Every order may be executed at the clearing price instead of their proposed price. Due to this rule, limit buying orders with a proposed price below the clearing price and limit selling orders with a proposed price above the clearing price would not be executed. Despite the difference between the two trading systems, a continuous limit order book can be considered a periodic auction whose duration equals $0$ second, see for example \cite{Jusselin}.

\subsection{Comparison and main flaws of limit order book}

The literature to promote general market quality, discover better trading mechanisms, or improve market competition has been studied since the 60s; see \cite{garbade}. The continuous trading system has the advantage of providing ''immediate execution''. No one likes to wait; in \cite{kalay}, it is shown empirically that people prefer to trade in a continuous market instead of an auction market. However, such immediacy also creates a problem, especially after the emergence of high-frequency traders. \cite{Budish}, \cite{wah2013}, and \cite{fs2012} question the efficiency of limit order mechanism rather than periodic auction. They study the efficiency of periodic auctions to monitor high-frequency trading advantages and increase market efficiency. \cite{Budish} compares two highly correlated stocks from real data and finds that the continuous trading system creates arbitrage opportunities in small time intervals. These arbitrage opportunities could be caught by high-frequency traders and thus incite competition in speed rather than price. High-frequency trading has brought down the execution time from several seconds at the start of the 2000s to microseconds nowadays. \cite{wah2013} reaches a similar conclusion as they use simulations to show that high-frequency traders are latency arbitrageurs and widen the bid-ask spread. \cite{fs2012} discusses the negative impacts of high-frequency trading and proposes using periodic auctions (they also propose pro-rata rules with continuous market and randomized auction duration with periodic auctions).

\bigskip

Following these works, more studies focus on the advantages and disadvantages of continuous limit order books and periodic auctions. \cite{Aquilina} use exchange message data to quantify the speed competition in \cite{Budish}. However, restricting the competition in speed is only one of the characteristics of the periodic auction system compared to the continuous system. Recall that the other characteristic of a periodic auction differing from the CLOB is that the clearing price is set by combining the opinions of a batch of orders instead of just two. Such characteristics view market supply and demand more comprehensively and thus could improve the price discovery process. \cite{Jusselin} shows that optimally setting a clearing rule (price discovery and auction duration) for the periodic auction system enables the clearing price of most assets to be closer to the efficient prices compared to the continuous limit order book system. However, the continuous system could sometimes be optimal regarding the above mentioned price discovery process. \cite{energy} shows by using real data that if we replace the continuous German Electricity Market with a frequent batch auction, there will be less traded volume but better price discovery (price is less noisy and closer to the fundamental value) and less liquidity cost measured in round-trip (CRT) cost. There are, of course, different opinions. In \cite{zhangib}, they show empirically that sub-second frequent batch auction leads to a decline in adverse selection cost but an increase in relative spread and a decrease in information efficiency measured by ''autocorrelation of midpoint returns''.  One thing to note about all these works is that researchers use different assumptions, models, and measures to reach their conclusions, so seemingly contrasting conclusions do not necessarily imply a contradiction.
%It is worth noting that the paper Jusselin et al. (2020) is original in emphasizing the importance of the choice of the right duration for a periodic auction design. Comparing a periodic auction system that is not equipped with the optimal duration with a continuous system undermines the true efficiency of a periodic auction system. 

\bigskip

Echoing the conclusion of \cite{Jusselin}, ``One size does not fit all''. Neither periodic auction nor continuous limit order book is the best by all measures and neither would benefit all affiliated groups. The interest of this paper is not to compare periodic auction and CLOB, but to study the possibility of a co-existence of the two systems. \cite{Derchu} proposes an \textit{Ad Hoc Electronic Auction Design} (AHEAD) which allows traders to switch between continuous trading sessions and periodic trading sessions. They show that this design enables a less volatile clearing price and traders especially the smaller players benefit from this design compared to a continuous system or a periodic auction system.

\bigskip

In addition to the CLOB and the periodic auction market, there have been focuses and advances in other trading mechanisms. Dark pools differ from CLOB and auction in that orders are not displayed to the public; see \cite{Ye11}, \cite{Zhudark}, \cite{Ye24}, \cite{BMMR} for studies on whether a dark pool harm or help with market efficiency. See \cite{Melton} for a latency floor design on CLOB to limit high frequency trading. Note also that cryptocurrency trading markets have interesting mechanism designs as well; see \cite{crypto24} for combining batch auction with automated market maker. 

\bigskip
We also want to mention additional works as for example \cite{du2017optimal}, \cite{Brinkman}, \cite{fricke2018too}, and \cite{Jusselin}. Each of these articles study the design of periodic auction to answer the following question: What is the optimal auction duration? \cite{wah2016} uses simulation to show how fast and slow traders would choose between continuous market and periodic auction if the two markets run together; the model in this work does not consider strategic timing though. \cite{dusize} proposes to add a size discovery market (''workup'') along with a batch auction market to increase allocation efficiency. A size discovery market allows traders to exchange inventory at a fixed-price so traders need not to worry about their price impact. The main differences between their model and ours are that they assume a strategic player in a batch auction submits a demand function instead of an order price or quantity and they focus on balancing the inventory level of each strategic player as they assume that a equal distribution of inventory among traders is the most desired. Despite the difference, their proposed design is worth to consider. Relating to \cite{dusize}'s concern, \cite{goldberg} uses a Nash equilibrium model to show that strategic players in an auction lower demand and supply to avoid moving the clearing price away from their interests and such behavior could lead to a loss of trading volume in the market.

\subsection{Optimal policies to cure auction's inefficiencies and related works}

In this paper, we are interested in furthering the design of AHEAD, the concurrence of a continuous market and a periodic auction market, and we focus specifically on the periodic auction part. We want to see whether a strategic trader could take advantage of the current setting when an auction is open. The main question raising our motivation is the following. 

\begin{center}
Would a strategic player disturbs market efficiency by strategically picking their arrival time? 

If so, how could we improve the auction design to bring back efficiency? 
\end{center}
\bigskip

The efficiency measure we use is a price discovery measure, the difference between the actual clearing price of an asset and its theoretical efficient price. To get a larger picture, ``market efficiency'' usually refers to either ``external efficiency'' or ``internal efficiency'' (see \cite{west}). A market is externally efficient if its prices reflect all available information. Fama develops this definition in \cite{fama} (also see \cite{Malkiel}). Relevant criteria to measure the market quality include auto-correlation of return, delay measure, and Sharpe ratio (see for example \cite{griffin} and \cite{sharpe}). Such efficiency is called \textit{external} because it depends not on trading systems but on the outside world, such as how information is spread among traders. Internal efficiency refers to whether a market could enable traders to trade at prices close enough to their desired prices. The difference could be generated by the distance between the executed price and the true price (also called pricing error, price discovery measure, or price formation measure) and by transaction costs such as liquidation and bid-ask spread. We adopt the price discovery measure because it suits our purpose the most and because most relevant work uses such measure (see all relevant work cited in the introduction and see \cite{madhavan1992trading}).

%wah and wellman defines price discovery as the distance between true fundamental price and executed price

\bigskip

The main results of our paper are as following. By studying the behavior of a strategic trader in a periodic auction, we emphasize that the trader benefit from arriving only at the end of an auction and right before the auction closes. Such strategic choice of arrival time enables the trader to take advantage of all known information to submit a strategic price and have the greatest price impact on the market. Not only is such strategic timing unfair to other market participants, but also the strategic pricing could drive the clearing price away from the efficient price of the underlying asset if the strategic trader has an incorrect guess about the efficient price. It can be seen as a disturbance to market efficiency.

\bigskip
We propose two regulatory policies to respond to this mechanical problem inherent in periodic auction markets. First, we introduce a randomization design of the auction's closing time. Second, we introduce a transaction fee indexed on the arrival time of every market participant in the auction, i.e., the later a trader arrives, the higher they pay. We prove that the randomization and the transaction fee design address the problem efficiently and bring better quality to the market. Note that the randomization of the closing time is mentioned in \cite{fs2012} and \cite{wah2016}. However, none of these works study the possible effects of such design and use quantification methods to address it. In reality, randomization of closing time is implemented, so it is possible to conduct an empirical analysis of such a design. London Stock Exchange adds a 30-second random period to the opening auction and closing auction. Cboe randomizes the whole duration of its periodic auction, i.e., an auction might close in 0 to 100 milliseconds. However, we caution that any conclusions on empirical analysis of real-world auction design might not apply to auction design because real-world auctions have specific settings, including auction duration and priority rules. Excluding or changing any of the settings could lead to a different conclusion.

\bigskip

\bigskip

Additionally, we would like to remind that our model of a periodic auction adopts the most basic set of rules. Periodic auction markets in reality could be much more complicated. Apart from Cboe's priority rule we mentioned above, the opening and closing auction of NYSE and Nasdaq each add their own matching and execution rules. These rules could possibly limit many advantages of a periodic auction market; see \cite{Jegadeesh} for an empirical analysis of NYSE and Nasdaq's closing auctions. 

\bigskip

The structure of this study is the following. In Section \ref{sec;auctionmod}, we introduce the auction mechanism and the modeling without fees or randomization. Section \ref{sec;auctionmod} presents the main charcateristics of the market, the mathematical model and the information set available for the strategic trader. Section \ref{sec:define} set the clearing price rule to trade the considered asset ensuring the largest number of matching orders (see Proposition \ref{prop:clearingp}). Section \ref{sec;stratopt} define the optimization problem of the strategic trader without regulation (randomization of the clearing time and transaction fees) together with its impact of the market quality. Section \ref{sec;data} introduce the data set together with the calibration of the relevant parameters for the study. Section \ref{sec;fullinfo} presents the solution of the problem when the strategic trader has a full information on the traded asset illustrated with numerical results. Section \ref{sec;inefficiency} studies the case where the strategic trader is imperfectly informed about the efficient price of the asset. Section \ref{sec;bilevel} turns to the impact of a clearing time randomization and transaction fees policy on the strategic trader behavior, the price impact and the market quality. The bilevel optimization is first introduced, following by the result considering only a randomization of the clearing time without fees then adding the fees. We consider two different problems from the exchange perspective to improve the market quality: either to reduce the price impact including fees paid by traders (Section \ref{sec;mq1}) or to reduce the distance between the clearing price and the efficient price while benefiting from the fees structure (Section \ref{sec:qualityfees}). Section \ref{sec;ccl} concludes the study and provides future perspective.

%(2)cboe does not display periodic auction order, but display the real time clearing price and quantity. Thus, our result would not be impaired too much, since we work mostly with sum P_i, not single pi

%(3) a critic could be that transaction cost might counter the improve of price discovery. after all, the clearing price is not the final executed price. defense could be: the sum is still better? (need to check), and the transaction cost is not higher than bid-ask spread in clob? so with transaction cost, at least better than clob 

\section{Auctions market modeling with transaction fees and randomization}\label{sec;auctionmod}

\subsection{The market characteristics}\label{sec:define}
We consider an auction to trade a risky asset starting at time $0$ with duration $T>0$. We denote by $P^{cl}_T$ the clearing price of the auction determined by the exchange to maximize the number of trades at the clearing time $T$. During the auction's duration, limit orders arrived such that each limit order $i$ is characterized by a limit price $P_i$ at which a trader is willing to buy or sell the asset and a volume determined by a supply function $Q_i = K(P^{cl}_T - P_i)$. The parameter $K>0$ is the slope of the supply function assumed to be fixed for each limit order. Note that linear supply functions are also considered in \cite{du2017optimal,fricke2018too,Jusselin} We assume that $(P_i)_{i\geq 1}$ is a family of independent normally distributed random variables with mean $\mu^{mm}$ and standard deviation $\sigma^{mm}$. Note that if $Q_i \leq 0$, the order is a buying order; if $Q_i \geq 0$, the order is a selling order. We assume that the efficient price of the risky asset denoted by $P^*$ is a normal random variable with mean $\mu^{*}$ and standard deviation $\sigma^*$. We assume $\mu^{mm} = \mu^{*}$ and $\sigma^{mm} = \sigma^* = \sigma$, for some $\sigma>0$. We model the arrival of these limit orders by a Poisson process $M$ with intensity $\lambda^{mm}_t=\lambda \times t$ where $\lambda$ is a positive constant. In other words, $N_t : = M_t + 1$ denotes the number of market makers active in the auction up to time $t$.\footnote{Here, "+1" practically means that there is at least one trader in the auction market for it to be open and theoretically to avoid division by zero. This assumption is consistent with the existence of liquidity in a CLOB transfered to an auction at for example the end of the day, see \cite{Jusselin}.} %We model this process as limit orders pre-existing in a CLOB, set by market makers and sent continuously in the coexistence system of CLOB and auction market. These limit orders have been set in the CLOB and are transferred to the auction for execution as block trades.\vspace{0.5em}

We define a family of $\sigma$-algebras $\mathcal{F} = \{\mathcal{F}_t\}_{0 \leq t \leq T}$ generated by the available information up to time $t$ and composed with the number of limit orders arrived in the auction and their limit prices $P_i$ that is $\mathcal{F}_t := \sigma\{N_t,(P_i)_{i=1}^{N_t}\}$.\vspace{0.5em}
%and the efficient price up to  time $t$,
%the filtration is right continuous after we define each element

We consider a strategic seller joining the auction at a deterministic time $\tau$ between 0 and $T$ aiming at optimally liquidating her position in the risky asset.\footnote{We could similarly assume that the strategic trader is a buyer but only consider the seller case motivated by optimal liquidation problem for the sake of simplicity.} The strategic seller controls the exact time $\tau$ she arrives in the auction together with the price $P$ at which she is willing to sell the asset, and the direction of the trade (which is selling). In this case, the volume sends by this strategic seller is $Q = K(P^{cl}_T - P)$. We assume that the price $P=P_\tau^\mu$, sent at time $\tau$ in the auction, is a normal distribution with mean $\mu$ where $\mu$ is a random variable controlled by the seller and variance $\sigma^2$ fixed measurable with respect to the information available for the trader. Note that while the variance for the market makers and the strategic seller are the same for the sake of simplicity, the strategic trader knows up to time $t$ the number of arrivals $N_t$ and the prices of these orders $\{P_i\}_{i=1}^{N_t}$. In particular, the strategic trader may be imperfectly informed about the efficient price of the asset since $\mu$ is determined through $\mathcal F$ which does not take into account the efficient price. The seller uses the information available at time $t$ to determine $\mu_t$, so $\mu$ could be viewed as a function $\mu_t = \mu(t, N_t,\{P_i\}_{i=1}^{N_t})$. By denoting $P^\mu_t$ the price the seller proposes entering in the auction at time $t$, we assume that $P_t$ is a normal distribution $\mathcal{N}(\mu(t, N_t,\{P_i\}_{i=1}^{N_t}),\sigma^2)$ where the function $\mu$ is controlled by the strategic seller when she enters at time $t$ and sees $N_t$ arrivals with associated price $\{P_i\}_{i=1}^{N_t}$.\vspace{0.5em}

Since the strategic seller controls the direction she trades, her order would not be executed if $P^{cl}_T < P^\mu_t$. However, since all other market makers do not control the directions of their trades, their orders will be executed for sure. The quantity the strategic seller trades is thus given by $K(P^{cl}_T - P^\mu_t)$ if $P^{cl}_T > P^\mu_t$, and $0$ otherwise. \vspace{0.5em}

The strategic trader does not necessarily know exactly the mean $\mu^*$ of the efficient price and the mean $\mu^{mm}$ of other transferred limit orders. We denote by $\mu_{g}^{*}$ and $\mu_{g}^{mm}$ the estimations the strategic trader of $\mu^*$ and $\mu^{mm}$ respectively. We denote by $\mathbb{E}_{g}$ the expectation when the mean of $P^*$ and $P_i$ corresponds to these estimations and we denote by $\mathbb{E}$ the expectation in the case $\mu_{g}=\mu^*$ and $\mu_{g}^{mm}=\mu^{mm}$.

\subsection{Clearing Price rule}

The clearing price of an auction is set to maximize trading volumes. When there is no strategic trader, corresponding to the case when no traders control the direction of the proposed prices, this clearing price is the equilibrium between supply and demand, see \cite{du2017optimal} or \cite[Section 2.3]{Jusselin}. That is 

%\footnote{We use $N_T+1$ instead of $N_T$ because practically an auction market has to have at least one trader to be open and mathematically to avoid being divide by $0$.}

\begin{equation}
    P^{cl}_T =  \frac{\sum_{i=1}^{N_T} P_i } {N_T},
    \label{Eqcl}
\end{equation}
which sets the clearing price to eliminate any imbalance between buy and sell orders.\vspace{0.5em}

Assume now that the strategic seller sends a price $P^\mu_t$ at time $t$ to trade the asset at the clearing time $T$ and will active only if the clearing price is above $P^\mu_t$, we have the following result.

%For example, consider a simple setting: if there are two buyers willing to buy at 10 and 12 respectively, and three sellers willing to sell at 8, 9, 11, respectively, the clearing price would be set between 9 and 10 so that four orders would be executed. If the clearing price is set below 9 or above 10, the total number of executed orders has to be strictly less than 4, which implies that the trading volume is not maximized. More interesting facts about how a clearing price is set in the real world could also be found in the paper by Comerton-Forde and Rydge (2006).
%Comerton-Forde and Rydge (2006)'s "principle 1, 2, 3" might be useful in clarifying the design of the periodic system.

\begin{proposition}\label{prop:clearingp}
The clearing price of the auction is determined by

\begin{equation}
P^{cl}_T =
\left\{ \begin{aligned} 
    \frac{\sum_{i=1}^{N_T} P_i + P^\mu_t } {N_T  +1} &\qquad \text{ if } \frac{\sum_{i=1}^{N_T} P_i}{N_T}  > P^\mu_t\\
    %\text{ if } P^{cl}_T > P
     \\
    \frac{\sum_{i=1}^{N_T} P_i } {N_T} &\qquad \text{ otherwise.}
\end{aligned} \right.
\label{eq:clear}
\end{equation}

\end{proposition}

\begin{proof}

The governing rule of setting a clearing price is that the price maximizes the traded volume. 

Assume first that $\frac{\sum_{i=1}^{N_T} P_i}{N_T}  > P^\mu_t$ then $\frac{\sum_{i=1}^{N_T} P_i + P_t } {N_T+1} > \frac{P_tN_T + P^\mu_t } {N_T+1}=P^\mu_t.$ If we set $P^{cl}_T = \frac{\sum_{i=1}^{N_T} P_i + P^\mu_t } {N_T+1}$, $P^\mu_t$ would be executed as a selling order. Then it follows from \eqref{Eqcl} that $P^{cl}_T$ is set as above.

Assume now $\frac{\sum_{i=1}^{N_T} P_i}{N_T}  \leq P^\mu_t$, there are two possible cases, either $P^{cl}_T < P^\mu_t$ or $P^{cl}_T \geq P^\mu_t$. If $P^{cl}_T < P^\mu_t$, then the strategic seller's order would not be executed anyway, so $P^{cl}_T = \frac{\sum_{i=1}^{N_T} P_i } {N_T}$ follows from \eqref{Eqcl}. If $P^{cl}_T \geq P^\mu_t$, number of executed buying orders:

\begin{align*}
    N^{buy} &= \min\big\{ K\sum_{i:P_i > P^{cl}_T } (P_i - P^{cl}_T), \quad K\sum_{i:P_i < P^{cl}_T } (P^{cl}_T - P_i) + K (P^{cl}_T - P^\mu_t) \big\} \\
    &= K\sum_{i:P_i > P^{cl}_T } (P_i - P^{cl}_T) \text{\quad because \quad} K\sum_{i=1}^{N_T} (P_i - P^{cl}_T) \leq KN_T(P^\mu_t-P^{cl}_T) \leq 0 \leq K (P^{cl}_T - P^\mu_t)
\end{align*}

However, $N^{buy} \leq K\sum_{i:P_i > \frac{\sum_{i=1}^{N_T} P_i}{N_T} } (P_i - \frac{\sum_{i=1}^{N_T} P_i}{N_T})$. This implies that traded volumes would be larger if the clearing price is set as $\frac{\sum_{i=1}^{N_T} P_i}{N_T}$. This violates the clearing price setting rule when $P^{cl}_T > P^\mu_t$. Thus, $P^{cl}_T \leq P^\mu_t$ and so $P^{cl}_T = \frac{\sum_{i=1}^{N_T} P_i } {N_T}$.
\end{proof}
The clearing price \eqref{eq:clear} could also be written as:

\begin{equation*}
    P^{cl}_T = \frac{\sum_{i=1}^{N_T} P_i + \mathbf 1_{\frac{\sum_{i=1}^{N_T} P_i}{N_T}  > P^\mu_t} P^\mu_t } {(N_T)+ \mathbf 1_{\frac{\sum_{i=1}^{N_T} P_i}{N_T}  > P^\mu_t}}.
\end{equation*}

The clearing price $P^{cl}_T$ depends on the strategic trader's input $P^\mu_t$. In the following, we would write $P_T^{cl}$ instead of $P_T^{cl}(P^\mu_t)$ for convenience; however, $P_T^{cl}$ is a function of $P^\mu_t$.

\subsection{Strategic Trader's optimization and market quality}\label{sec;stratopt}

\subsubsection{Strategic trader optimization}
The strategic seller sends at time $\tau$ a volume $K(P_T^{cl}-P^\mu_{\tau})$ in the auction to sell the asset under the condition $P^{cl}_T > P^\mu_\tau$. If executed, the value of the strategic seller's portfolio at time $T$ is compared with the efficient price. The seller's payoff at the clearing time is thus $K(P^{cl}_T - P^\mu_{\tau})(P^{cl}_T - P^{*})$ if $P^{cl}_T > P_\tau$, and $0$ otherwise. \vspace{0.5em}

For any $(t,n,(p_i)_{i=1}^n) \in [0,T]\times \mathbb N\times \mathbb R^{n}$, we define
\begin{equation}
    \hat{\mu}(t,n,(p_i)_{i=1}^n) := \arg\max_{\mu} \mathbb{E}_{g}[\mathbf 1_{p_{\mu} \leq P_T^{cl}} K(P^{cl}_T - p_{\mu})(P^{cl}_T - P^{*}) | N_{t} = n, (P_i)_{i=1}^{N_t} = (p_i)_{i=1}^n  ],
    \label{eq:choosemu_basic}
\end{equation}
where $p_{\mu}$ refers to a normal random variable $\mathcal{N}(\mu,\sigma^2)$.

\begin{remark}
    In the proof of Theorem \ref{thm:notrader} and remark \ref{remark:hatmu} below, we will show the uniqueness of $\hat{\mu}$. However, if equation \eqref{eq:choosemu_basic} outputs more than one $\arg\max$ value, we define $\hat{\mu}$ as the minimum of these outputs.
\end{remark}

We set 
\begin{align*}
    &\hat{\mu}_{\tau} := \hat{\mu}(\tau,N_\tau,(P_i)_{i=1}^{N_\tau}),\\
    &\hat{P}_{\tau}=P_\tau^{\hat\mu} \sim \mathcal{N}(\hat{\mu}_{\tau},\sigma^2).
\end{align*}

Since the strategic seller wants to maximize payoff, the seller's problem is written as:
\begin{align}
\begin{split}
   V^{\circ}=\sup_\tau\; V^{\circ}(\tau), 
    \label{Eq:discretetrader}
\end{split}
\end{align}
with \[ V^{\circ}(\tau)= \mathbb{E}_{g}\Big[\mathbf 1_{\hat {P}_\tau \leq P_T^{cl}} \big\{ K(P^{cl}_T - \hat P_\tau)(P^{cl}_T - P^{*})  \big\}\Big]. \]

%\footnote{ Set $(\Omega, \mathcal{G},\mathbb{P})$ to be the probability space on which $\{N_t,(P_i)_{i=1}^{N_t}, P^*\}$ and a Brownian motion $W_t$ is defined. For a $\mathcal{F}$-adapted process $ \mu'$, define $\mathbb{P}^{\mu'}$ by \cite{shreve}: 

%$$\frac{d\mathbb{P}^{{\mu'}}}{d \mathbb{P}}|_{\mathcal{G}} = \exp\{\int_0^T  \mu'_tdW_t - \frac{1}{2} \mu'_t^2 dt \}$$.

%Under $\mathbb{P}^{\mu'}$, $W_{\tau}  \sim \mathcal{N}({\mu'}_{\tau},\tau)$. By scaling, we get ${P}_{\tau} \sim \mathcal{N}({\mu'}_{\tau},\sigma^2)$ under $\mathbb{P}^{\mu'}$. We introduce $\mathbb{E}^{\mu}$ to denote the expectation with regard to measure $\mathbb{P}^{\mu'}$.\vspace{0.5em}

%In the following paper, to simplify notations and improve readability, we refer to $\mathbb{E}^{\mu'}$ and $\mathbb{P}^{\mu'}$ whenever we mention $\mathbb{E}[fct(P_{\mu'})]$ or $\mathbb{P}[fct(P_{\mu'}]$.}

Intuitively, the optimal arrival of the trader should be always $T$ as the trader could use more information to make decisions. The following theorem proves that the optimal time to arrive is indeed $\tau =T$; the later the trader joins the auction the better the payoff. 

\begin{theorem}
The strategic trader benefits from arriving as late as possible in an auction. In other words, for any $0\leq s\leq t\leq T$ we have $V^\circ(s) \leq V^{\circ}(t)$.
\label{thm:discretetra}
\end{theorem}

\begin{proof}
    We have
    \begin{align*}
        V^{\circ}(s) &=   \mathbb{E}_{g}\Big[\mathbf 1_{\hat{P}_{s} \leq P_T^{cl}} K(P^{cl}_T - \hat{P}_{s})(P^{cl}_T - P^{*}) \Big] \\
        &=  \mathbb{E}_{g}\Big[ \mathbb{E}_{g}\big[\mathbf 1_{\hat{P}_{s} \leq P_T^{cl}}  K(P^{cl}_T - \hat{P}_{s})(P^{cl}_T - P^{*})  |\mathcal{F}_s \big]\Big] \\
        &= \mathbb{E}_{g} \Bigl\{ \mathbb{E}_{g}\Big[ \mathbb{E}_{g}[\mathbf 1_{\hat{P}_{s} \leq P_T^{cl}} K(P^{cl}_T - \hat{P}_{s})(P^{cl}_T - P^{*})   |\mathcal{F}_t ] \Big|\mathcal{F}_s \Big] \Bigl\}
    \end{align*}
    
    where the last equality is based on $\mathcal{F}_s \subset \mathcal{F}_t$.

    Similarly,
        \begin{align*}
        V^{\circ}(t) &= \mathbb{E}_{g} \Bigl\{ \mathbb{E}_{g}\Big[ \mathbb{E}_{g}[\mathbf 1_{\hat{P}_{t} \leq P_T^{cl}} K(P^{cl}_T - \hat{P}_{t})(P^{cl}_T - P^{*})   |\mathcal{F}_t ] \Big|\mathcal{F}_s \Big] \Bigl\}.
    \end{align*}
    
    As $\hat{\mu}_{s}$ is $\mathcal{F}_s$-measurable and $\mathcal{F}_s \subset \mathcal{F}_t$, $\hat{\mu}_{s}$ is $\mathcal{F}_t$-measurable. Then by the definition of $\hat\mu$ in \eqref{eq:choosemu_basic}, \[\mathbb{E}_{g}\big[\mathbf 1_{\hat{P}_{t} \leq P_T^{cl}}  K(P^{cl}_T - \hat{P}_{t})(P^{cl}_T - P^{*}) |\mathcal{F}_t \big] \geq \mathbb{E}_{g}\big[\mathbf 1_{\hat{P}_s \leq P_T^{cl}} K(P^{cl}_T - \hat{P}_s)(P^{cl}_T - P^{*})| \mathcal{F}_t \big],\; \mathbb P-a.s..\] Then $V^\circ(t) \geq V^{\circ}(s)$.  
\end{proof}

\subsubsection{Market quality and exchange's viewpoint}
While the strategic trader wants to maximize her payoffs, the exchange would benefit from an arrival of the strategic trader which minimizes the spread between the clearing price $P^{cl}_T$ and the efficient price $P^{*}$. We denote by $P^{cl,\tau}_T$ the clearing price given by \eqref{eq:clear} for $P=\hat P_\tau$, that is the clearing price if the strategic seller arrives at time $\tau$. We have
\begin{equation*}
P^{cl,\tau}_T =
\left\{ \begin{aligned} 
    \frac{\sum_{i=1}^{N_T} P_i + \hat P_\tau } {N_T + 1} &\qquad \text{ if } \frac{\sum_{i=1}^{N_T} P_i}{N_T}  > \hat P_\tau\\
    %\text{ if } P^{cl}_T > P
     \\
    \frac{\sum_{i=1}^{N_T} P_i } {N_T} &\qquad \text{ otherwise.}
\end{aligned} \right.
\end{equation*}

We introduce two different measures (disutility functions) of market quality: 

\begin{align*}
    (MQ)(\tau): &= \mathbb{E}[ |P^{cl,\tau}_T - P^{*}|^2 ], \\
    (MQ)^{\rho}(\tau): &= \mathbb{E}[ \exp{(\rho|P^{cl,\tau}_T - P^{*}|)} ],
\end{align*}

where $\rho>0$ is the risk aversion of the exchange with respect to the spread. As a benchmark and first-best case scenario, we assume that the exchange controls the arrival of the strategic seller. The exchange aims at solving

%In other words, they want to limit the manipulation power of a strategic trader who could make $P^{cl}_T$ deviate from the efficient price far enough to favor himself. 

\begin{equation}
     \min_{\tau_{reg}} MQ (\tau_{reg}) = \min_{\tau_{reg}} \mathbb{E}[ |P^{cl,\tau_{reg}}_T - P^{*}|^2 ],  \label{eq:discretereg}
\end{equation}
or
\begin{equation}
     \min_{\tau_{reg}} MQ^{\rho} (\tau_{reg})  = \min_{\tau_{reg}} \mathbb{E}[ \exp{(\rho|P^{cl,\tau_{reg}}_T - P^{*}|)} ]. \label{eq:discretereg:rho}
\end{equation}

%where $\tau_{reg}$ is the optimal arrival time of the strategic trader that favors by the exchange.

%\subsection{Inefficiency without monitoring policies: data and numerical analysis}\label{sec;inefficiency}
\subsection{Data and numerical analysis}\label{sec;data}

We now investigate the optimal arrival of the strategic seller solving \eqref{Eq:discretetrader} together with the optimal deviation form the efficient price $\hat \mu$ proposed in the auction and the market quality given by \eqref{eq:discretereg} or \eqref{eq:discretereg:rho}. We refer to Appendix \ref{appendix:discrete} for the details of the computations to perform the numerical study studied in this section.\vspace{0.5em}

We set $T = 10$ and discretize the time span by assuming that traders only join the markets at time $\tau \in \{1,2,...,9, 10\}$. We use trading data extracted from YahooFinance for Apple and Alphabet (Google) stock on period Oct-2-2023 to Dec-29-2023 to calibrate the parameters $\mu^{mm}$ and $\sigma$. We consider three months' data to avoid being affected by any single period's abnormal behavior.\vspace{0.5em}

We set the mean $\mu^{mm}$ of $P^{*}$ and $\{P_j\}_{j}$ to be the average day price of the three month period considered. For each trading day, we compute each stock's day price by (Open Price + Close Price + High Price + Low Price)/4. We get Apple's $\mu^{mm} = 184.39$ and Alphabet's $\mu^{mm} = 134.24$.\vspace{0.5em} 

As for the standard deviation $\sigma$ of $P^{*}$ and so $P_t$ and $\{P_j\}_{j}$, we use the formula 
\[\sigma^2 = \frac{1}{N}\sum_{k=1}^N (S_k - S_{k-1})^2,\] where $S_k$ and $S_{k-1}$ are the price of the stock day $k$ and day $k-1$ respectively, $N$ denotes the total number of days of the period. We got from the data set Apple's $\sigma = 1.76$ and Alphabet's $\sigma = 2.11$.\vspace{0.5em}

When the strategic seller is imperfectly informed about the drift of efficient price and the drift of the price proposed by the other traders, the estimation of $\mu^*$ and $\mu^{mm}$ by the strategic seller, we work under two scenarios. We first consider the case where the strategic seller under-estimates these parameters and named this case Case $(-)$, that is $\mu_{g}=\mu^*-\sigma$ and $\mu_{g}^{mm}=\mu^{mm}-\sigma$. Symmetrically, we consider the case where the strategic seller over-estimates $\mu^*$ and $\mu^{mm}$ and named this case Case $(+)$, that is $\mu_{g}=\mu^*+\sigma$ and $\mu_{g}^{mm}=\mu^{mm}+\sigma$.\vspace{0.5em}

We further set a bound for the strategic trader's $\hat{\mu}$ in the bounded set $[\mu_{g}^{*} - 4\sigma,\mu_{g}^{*} + 4\sigma]$.\footnote{Note that this assumption is not unrealistic in view of the existing literature since most of market making mathematical models require to chose a spread in a bounded interval. Moreover, in remark \ref{remark:hatmu}, we see that $\hat{\mu}$ could go to infinity for some cases, so it is necessary to bound $\hat{\mu}$ when solving numerically.} We assume that the slope of the supply function is $ K = 10$ and take $\lambda = 1$ in consideration of the computation cost.\vspace{0.5em}

\subsection{Strategic trader with full information: efficient but unfair market}\label{sec;fullinfo}

In a perfect world setting, $\mu^*_{g} = \mu^*$ and $\mu^{mm}_{g} = \mu^{mm}$, \textit{i.e.} the strategic trader either guesses correctly or has insider information. 

\subsubsection{The stock exchange prefers a strategic trader to join the market}\label{sec;prefertrader}

In this subsection, we emphasize the benefit for the exchange to attract a strategic trader in the market. We modify the assumption in section \ref{sec:define} by assuming $\mathcal{F}_T = \sigma\{N_T, (P_i)_{i=1}^{N_T},P^*\}$. Note that such modification does not change the result of Theorem \ref{thm:discretetra} which says that a strategic trader always join at $\tau = T$. We define $P_{T}^{cl,\emptyset}$ to be the clearing price in an auction market when there is no strategic trader. According to \eqref{eq:clear}, we have 
\[P_{T}^{cl,\emptyset} = \frac{\sum_{i=1}^{N_T} P_i } {N_T}.\] Hence, the market quality when no strategic trader arrives in the auction is given by 
\[MQ^{\emptyset} = \mathbb{E}[ |P^{cl,\emptyset}_T - P^{*}|^2 ].\]

\begin{remark}
    Without a strategic trader, the exchange would prefer more limit orders in the auction. This is because $MQ^{\emptyset} = \sigma^2(1+\frac{1-e^{-T\lambda}}{T\lambda})$ which decreases monotonically with $\lambda$.
    \label{remark:moretrade}
\end{remark}

\begin{theorem}
    Assume that the strategic trader is either a seller or a buyer, \textit{i.e.}, the trader's objective is defined as 
    \begin{align}
        \label{pb:gen}
        \sup_{\tau}\mathbb{E}[  K(P^{cl}_T -  \overline \mu_{\tau})(P^{cl}_T - P^{*})],
    \end{align}
    where $\overline {\mu}_t$ is the optimizer for $\esssup_{\mu} \mathbb{E}[  K(P^{cl}_T -\mu)(P^{cl}_T - P^{*})|\mathcal{F}_t]$. Then $MQ(\hat{\tau}) =  \frac{1}{4} MQ^{\emptyset} < MQ^{\emptyset}$ where $\hat \tau$ is the optimal arrival time of the strategic trader optimizing \eqref{pb:gen}. In other words, the exchange prefers the arrival of a strategic trader in the auction to improve the market quality.  
        \label{thm:notrader}
\end{theorem}

%$\arg\max_{\mu} \mathbb{E}[  K(P^{cl}_T - p_{\mu})(P^{cl}_T - P^{*}) | N_{t}, (P_i)_{i=1}^{N_t}  ]$

\begin{proof}
%Fix $\omega \in \mathcal{F}_T$, the strategic trader's problem is to find the optimizer $\bar{\mu} (\omega)$ for $\sup_{\mu} \mathbb{E}^{\mathcal{F}_T} \big[\big\{ K(P^{cl}_T - \mu)(P^{cl}_T - P^{*})  \big\} \big] (\omega)$$
    
    Since $\hat{\tau} = T$, the strategic trader's problem is to find the $ \mathcal{F}_T-$measurable optimizer $\bar{\mu}$ for
    \begin{align}
    \begin{split}
    &\esssup_{\mu} \mathbb{E} \big[\big\{ K(P^{cl}_T - \mu)(P^{cl}_T - P^{*})  \big\}\big| {\mathcal{F}_T} \big] \\
    =&  \esssup_{\mu} \mathbb{E} \big[ \big\{ K(\frac{\sum_{i=1}^{N_T} P_i + \mu } {N_T + 1}- \mu)(\frac{\sum_{i=1}^{N_T} P_i + \mu } {N_T + 1} - P^{*})  \big\} \big| {\mathcal{F}_T}\big]\\
    =&  \esssup_{\mu} \mathbb{E} \big[\big\{ \frac{-N_T}{(N_T+1)^2}\mu^2 + \mu( \frac{(1-N_T)\sum_{i=1}^{N_T} P_i}{(N_T+1)^2}  + \frac{N_T}{N_T+1}P^*) + \frac{\sum_{i=1}^{N_T} P_i } {N_T + 1}(\frac{\sum_{i=1}^{N_T} P_i } {N_T + 1} - P^{*}) \big\}\big| {\mathcal{F}_T} \big].
     \end{split}
     \label{eq:noind}
    \end{align}

    The function $f$ defined by
    \begin{align*}
    &f(x)\\
    &=\mathbb{E}\left[\frac{-N_T}{(N_T +1)^2}\big| \mathcal{F}_T\right] x^2 + \mathbb{E}\left[ \frac{(1-N_T)\sum_{i=1}^{N_T} P_i}{(N_T+1)^2} + \frac{N_T}{N_T+1}P^*\big| \mathcal{F}_T\right] x + \mathbb{E}\left[ \frac{\sum_{i=1}^{N_T} P_i } {N_T + 1}(\frac{\sum_{i=1}^{N_T} P_i } {N_T + 1} - P^{*}) \big| \mathcal{F}_T \right] \\
     &= \left[\frac{-N_T}{(N_T +1)^2}\right] x^2 + \left[ \frac{(1-N_T)\sum_{i=1}^{N_T} P_i}{(N_T+1)^2} + \frac{N_T}{N_T+1}P^*\right] x + \left[ \frac{\sum_{i=1}^{N_T} P_i } {N_T + 1}(\frac{\sum_{i=1}^{N_T} P_i } {N_T + 1} - \mu^{*})  \right] 
    \end{align*}
    is maximized at $$x = -\frac{\frac{( N_T(N_T +1)P^{*} - (N_T -1)\sum_{i=1}^{N_T} P_i)}{(N_T +1)^2}}{2\frac{-N_T}{(N_T +1)^2}}.$$
    By the symmetry of the parabola $f$, the optimizer is $$\bar{\mu} = -\frac{\frac{( N_T(N_T +1)P^{*} - (N_T -1)\sum_{i=1}^{N_T} P_i)}{(N_T +1)^2}}{ 2\frac{-N_T}{(N_T +1)^2}}.$$

    If the strategic trader sends price at $\bar{\mu}$, then $$P_T^{cl}(\bar \mu) - P^* = \frac{1}{2}( \frac{\sum_{i=1}^{N_T} P_i}{N_T} - P^*) = \frac{1}{2}( P_{T}^{cl,\emptyset} - P^*).$$ Consequently,

    \[ \mathbb E[|P_T^{cl}(\bar \mu) - P^*|^2] =\frac14 \mathbb E[|P_{T}^{cl,\emptyset} - P^*|^2]=\frac14 MQ^\emptyset. \]

\end{proof}

\begin{remark}\label{remark:hatmu}
    In this remark, we illustrate where $\hat{\mu}$ is achieved. Since $\hat{\tau} = T$, the strategic seller's problem in \eqref{Eq:discretetrader} is to find the optimizer $\hat{\mu}$ for
    $$\esssup_{\mu} \mathbb{E} \big[\mathbf 1_{P^{\mu}_T \leq P_T^{cl}} \big\{ K(P^{cl}_T - P^{\mu}_T)(P^{cl}_T - P^{*})  \big\} \big|\mathcal{F}_T  \big].$$
    The difference between the above problem and problem \eqref{eq:noind} is the presence of the indicator function $\mathbf 1_{P^{\mu}_T \leq P_T^{cl}}$ and the trader sends $P^\mu_T$ centering around $\mu$ instead of sending $\mu$. If remove the indicator function from the above problem, we have $\hat{\mu} = \bar{\mu}$ by the symmetry of the probability distribution of a normal random variable.\vspace{0.5em} (1) Suppose for the chosen event $\omega \in \mathcal{F}_T$, $( P_{T}^{cl,\emptyset} - P^*) \geq 0$. By Theorem \ref{thm:notrader}, $\hat{\mu}(\omega) \geq \bar{\mu}(\omega)$ due to the presence of the indicator function $\mathbf 1_{P^{\mu}_T \leq P_T^{cl}}$ which shifts $\hat{\mu}(\omega)$ upward (see appendix \ref{appendix:illustrate} for more details).\vspace{0.5em} (2) Suppose for the chosen event $\omega \in \mathcal{F}_T$, $(P_{T}^{cl,\emptyset} - P^*) < 0$. For any given $\mu$, when $P^{\mu}_T > P_T^{cl}$, $\mathbf 1_{P^{\mu}_T \leq P_T^{cl}} K(P^{cl}_T - P^{\mu}_T)(P^{cl}_T - P^{*}) $ = 0 and when $P^{\mu}_T \leq P_T^{cl}$, $\mathbf 1_{P^{\mu}_T \leq P_T^{cl}} K(P^{cl}_T - P^{\mu}_T)(P^{cl}_T - P^{*})  = K(P^{cl}_T - P^{\mu}_T)(P^{cl}_T - P^{*}) 
    \leq K(P_T^{cl} - P^{\mu}_T)(P_{T}^{cl,\emptyset} - P^*) <0$. Thus, to maximize the objective \eqref{Eq:discretetrader}, the strategic trader would pick $\hat{\mu}(\omega)$ as large as possible, $\hat{\mu}(\omega) = \infty$.
    
\end{remark}

\subsubsection{Numerical results and economical insights}

The numerical analysis is presented in Table \ref{table:zerofee_perfect}. We find the optimal $\hat \tau_{reg} = 10$ by solving \eqref{eq:discretereg} or \eqref{eq:discretereg:rho}, which implies that the regulator benefits from an arrival of a strategic trader with full information at the end of an auction. Figure \ref{graph:rhoMQ} shows that the absolute value of the slope of $MQ^{\rho}$ increases when $\rho$ gets larger. This implies that the exchange would prefer a late-arriving trader in an auction even more if the exchange is highly risk averse (higher $\rho$) on market quality spread. In addition, we find the optimal $\hat \tau = 10$ for the arrival of the strategic seller solving \eqref{Eq:discretetrader} which confirms Theorem \ref{thm:discretetra} that traders prefer to arrive at the end.\vspace{0.5em}

The result $\hat \tau_{reg} = \hat \tau =10$ confirms the ``Law of One Price'' which says that arbitragers brings price convergence to the efficient price.\footnote{Suppose the efficient price is \$5 per share and the market price is \$4 share, then arbitragers would buy at \$4 to sell later at \$5. As they buy, they would push the market price upward until it reaches \$5 at which arbitragers stop buying thus stop moving the market price.} From this law, the exchange would prefer the strategic trader to arrive at a time that maximizes their arbitrage impact which is $\tau = 10$ when the trader knows everything about the market. In Table \ref{table:zerofee_perfect}, two columns show that the strategic trader's arbitrage impact increases as $\tau$ increases: column `Strategic Seller' shows that the seller's expected payoff increases with the arrival time and column `Price Impact' shows that the trader moves the clearing price more away from other transferred limit orders' aggregated opinion the later the trader joins the auction.\footnote{Price Impact :$= \mathbb{E}[|P_{T}^{cl,\emptyset}-P_{T}^{cl}|^2]$}  From the last column, we could see that the strategic trader exerts price impacting power by proposing higher prices as $\tau$ increases. Note that the proposed price is always higher than the true mean of the efficient price of underlying assets, implying that the seller wants to sell high to gain profits; in other words, the later the trader arrives, the larger the spread they proposes. \vspace{0.5em}

The agreement between $\hat \tau_{reg} = \hat \tau$ seems satisfying but indulging the strategic trader to join only at $\tau =10$ has several concerns. The strategic trader obviously uses time as a lever to gain advantage of earlier joined traders. Recall that our model of an auction market is a zero-sum game, so the greater gain of the trader, the larger loss of the other traders. This information unfairness could result in general traders losing interest in the auction market and turning entirely to the continuous market. In addition, the price impact of the strategic trader is a concern. From the table, we see the trader has the largest price manipulation power at $\tau =10$. This does not produce a problem in a perfect world with full information but is unrealistic and leads to flaws in a less perfect world as we discuss in the next section. 

\begin{table}[htbp]
	\caption{Results of Apple and Alphabet}
        \centering
	\label{table:zerofee_perfect}
	\begin{tabular}{@{}c|cccccc@{}}
		\toprule
            \textbf{Arrival Time of}& \multicolumn{5}{c}{\textbf{Apple}} &  \\
		\textbf{Strategic Seller}& Strategic Seller \textbf{\eqref{Eq:discretetrader}} & $MQ$ & $MQ^{0.1}$ & Price Impact & $\mathbb{E}[\hat\mu]$  \\
		\midrule
		  t=1 & 	0.5135 & 	3.2592 & 	1.1509 & 	0.0561 & 	185.2150\\
            t=2 & 	0.7489 & 	3.2516 & 	1.1509 & 	0.0721 & 	185.2238 \\
            t=3 & 	0.9871 & 	3.2435 & 	1.1506 & 	0.0874 & 	185.2996 \\
            t=4 & 	1.2273 & 	3.2354 & 	1.1503 & 	0.1022 & 	185.3871 \\
            t=5 & 	1.4691 & 	3.2274 & 	1.1501 & 	0.1166 & 	185.4779 \\
            t=6 & 	1.7123 & 	3.2195 & 	1.1499 & 	0.1307 & 	185.5669 \\
            t=7 & 	1.9566 & 	3.2117 & 	1.1496 & 	0.1444 & 	185.6473 \\
            t=8 & 	2.2014 & 	3.2035 & 	1.1493 & 	0.1586 & 	185.7147 \\
            t=9 & 	2.4451 & 	3.1948 & 	1.1487 & 	0.1731 & 	185.7359 \\
            t=10 & 	2.6843 & 	3.1832 & 	1.1471 & 	0.1869 & 	185.5862 \\
		\bottomrule
	\end{tabular}
\end{table}
\begin{table}[htbp]
        \centering
	\begin{tabular}{@{}c|cccccc@{}}
		\toprule
            \textbf{Arrival Time of}& \multicolumn{5}{c}{\textbf{Alphabet}} &  \\
		\textbf{Strategic Seller}& Strategic Seller \textbf{\eqref{Eq:discretetrader}} & $MQ$ & $MQ^{0.1}$ & Price Impact & $\mathbb{E}[\hat\mu]$  \\
		\midrule
		  t=1 & 	0.7381 & 	4.6863 & 	1.188 & 	0.0879 & 	135.6279\\
            t=2 & 	1.0765 & 	4.6755 & 	1.1879 & 	0.1134 & 	135.7143\\
            t=3 & 	1.4188 & 	4.6641 & 	1.1876 & 	0.1371 & 	135.829\\
            t=4 & 	1.764 & 	4.6525 & 	1.1872 & 	0.1602 & 	135.9517\\
            t=5 & 	2.1116 & 	4.6412 & 	1.1869 & 	0.1819 & 	136.0635\\
            t=6 & 	2.4612 & 	4.6299 & 	1.1866 & 	0.2035 & 	136.1738\\
            t=7 & 	2.8123 & 	4.6186 & 	1.1863 & 	0.225 & 	136.2715\\
            t=8 & 	3.164 & 	4.607 & 	1.186 & 	0.2467 & 	136.3597\\
            t=9 & 	3.5143 & 	4.5945 & 	1.1853 & 	0.2696 & 	136.4248\\
            t=10 & 	3.8581 & 	4.5779 & 	1.1836 & 	0.2907 & 	136.3411\\
		\bottomrule
	\end{tabular}
\end{table}
\begin{figure}[H]
	\centering
	\includegraphics[width=0.45\textwidth]{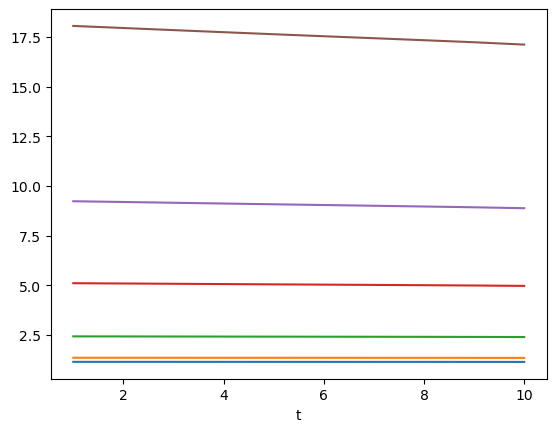}
        \includegraphics[width=0.45\textwidth]{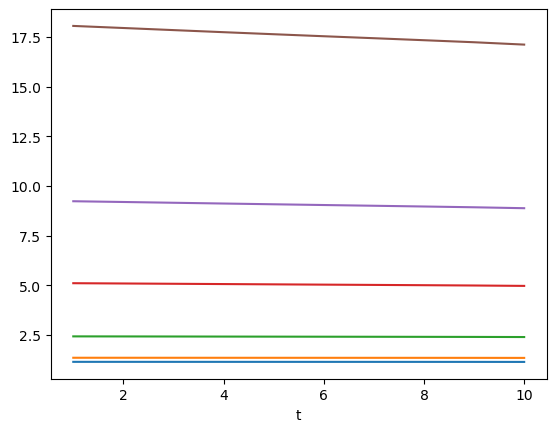}
         \caption{
         {\normalfont \small Apple's (left) and Google's (right) $MQ^{\rho}$. From the bottom curve to the top curve, $\rho$ fixed at $0.1, 0.2, 0.5, 0.8, 1, 1.2$ respectively. For each curve, it presents $MQ^{\rho}(t)$, with $t = 1, 2, ..., 10$. }}\label{graph:rhoMQ}
\end{figure}

\subsection{Imperfect information and inefficiency of auctions}
\label{sec;inefficiency}

In a more realistic framework, the strategic trader has a misconception of $\mu^*$ and $\mu^{mm}$. In the previous case, the trader uses price impacting power to drag clearing price towards the efficient price, but when the trader has a different target than the efficient price, the story changes. \vspace{0.5em}

We test two cases: case $(-)$ when $\mu^*_{g} = \mu^* - \sigma$ and $\mu^{mm}_{g} = \mu^{mm} - \sigma$ and case $(+)$ when $\mu^*_{g} = \mu^* + \sigma$ and $\mu^{mm}_{g} = \mu^{mm} + \sigma$.\footnote{Note that Case $(-)$ is more likely to happen than Case $(+)$ since we are studying a seller not a buyer.}\vspace{0.5em} 

The numerical analysis is presented in Table \ref{table:zerofee_imperfect}.\footnote{Alphabet's graphs carry the same spirit as Apple's graphs. For the sake of simplicity, we only analyze the data from Apple in this part.} In the case $(-)$ we find the optimal $\hat \tau_{reg} = 1$ by solving \eqref{eq:discretereg} or \eqref{eq:discretereg:rho}, which implies that the regulator wants the strategic trader to arrive at the beginning of an auction. Under this case, we see that the optimal arrival time of a strategic trader disturbs market efficiency as  $\hat \tau  =10$. In addition, we observe that the market quality becomes worse as $\tau$ increases, while the strategic seller's payoff increases with $\tau$. Thus the exchange would prefer the strategic trader to arrive as early as possible to minimize the spread between the efficient price and the clearing price, while the trader prefers to arrive as late as possible to maximize expected payoff. \vspace{0.5em} 

In the case $(+)$, we find the optimal $\hat \tau_{reg} = 1$ when using measure $MQ^{\rho}$ with large $\rho$, such $\rho =1, \rho =1.2$ and $\hat \tau_{reg} = 10$ when using measure $MQ$ or $MQ^{\rho}$ with small $\rho$. It is important to note that the difference between $MQ(t)$ or $MQ^{\rho}(t)$ across different $t$ is not significant, see Figure \ref{figure:imperfect_upper}. A possible explanation is that the trader tends to propose a high price due to the misconception $\mu^*_{g} = \mu^* + \sigma$, which results in many of the high-price orders not getting executed due to the fact that $P_T^{cl} < P$. As the trader's price manipulation power gets restricted in this way, the market quality measure varies little across different $\tau$ (the arrival time of the strategic trader). Following the same logic, the reason why we see a divided result (i.e., what is $\hat \tau_{reg}$) between large $\rho$ and small $\rho$ is that when $\rho$ grows the exchange becomes extremely sensitive to any possible difference between the clearing price and efficient price. Thus small chance events get exaggerated; for example, the event when $P_T^{cl,\empty}$ is high above $P^*$ and the strategic trader manages to propose a high $P$ such that $P < P_T^{cl}$. The strategic trader would be able to impact the price and catch this opportunity more if the trader arrives at $\tau=10$ instead of $\tau=1$, resulting $MQ^{\rho} (\tau =10) > MQ^{\rho} (\tau =1)$. In addition, note that $MQ^{\rho} (\tau =4) > MQ^{\rho} (\tau =10)$ even though the trader has larger price impact at $\tau =10$. Recall Section \ref{sec;prefertrader}, the exchange prefers a trader to join an auction market. The trader has more information to enable a successful execution (i.e. $P < P_T^{cl}$ ) at $\tau = 10$ than $\tau =4$ thus resulting this relation. Overall, case one deserves more attention than case two because case two misconception self-restricts a strategic seller's price impacting power.

\begin{table}[htbp]
	\caption{Results of Apple}
        \centering
	\label{table:zerofee_imperfect}
	\begin{tabular}{@{}c|cccccc@{}}
		\toprule
            \textbf{Arrival Time of}& \multicolumn{5}{c}{\textbf{Case ($-$): Lower Conjecture}} &  \\
		\textbf{Strategic Seller}& $MQ$ & $MQ^{0.1}$ & $MQ^{1}$ & Price Impact & $\mathbb{E}[\hat\mu]$  \\
		\midrule
		  t=1 & 	3.2765 & 	1.1514 & 	9.3008 & 	0.0747 & 	182.222\\
            t=2 & 	3.2963 & 	1.1521 & 	9.3864 & 	0.1212 & 	181.1093\\
            t=3 & 	3.3243 & 	1.1529 & 	9.512 & 	0.1728 & 	180.1473\\
            t=4 & 	3.3596 & 	1.1538 & 	9.6727 & 	0.2279 & 	179.2896\\
            t=5 & 	3.4 & 	1.1549 & 	9.8582 & 	0.283 & 	178.5356\\
            t=6 & 	3.4445 & 	1.1561 & 	10.0647 & 	0.3374 & 	177.8721\\
            t=7 & 	3.4898 & 	1.1574 & 	10.2765 & 	0.388 & 	177.3117\\
            t=8 & 	3.5356 & 	1.1586 & 	10.4933 & 	0.4351 & 	176.8247\\
            t=9 & 	3.5777 & 	1.1594 & 	10.6966 & 	0.4763 & 	176.3853\\
            t=10 & 	3.6161 & 	1.1591 & 	10.8899 & 	0.5133 & 	175.8701\\
		\bottomrule
	\end{tabular}
\end{table}
\begin{table}[htbp]
        \centering
	\begin{tabular}{@{}c|cccccc@{}}
		\toprule
            \textbf{Arrival Time of}& \multicolumn{5}{c}{\textbf{Case ($+$): Higher Conjecture}} &  \\
		\textbf{Strategic Seller} & $MQ$ & $MQ^{0.1}$ & $MQ^{1}$ & Price Impact & $\mathbb{E}[\hat\mu]$  \\
		\midrule
		  t=1 & 	3.2763 & 	1.1514 & 	9.3202 & 	0.1737 & 	188.0503\\
            t=2 & 	3.2772 & 	1.1516 & 	9.324 & 	0.2367 & 	188.8119\\
            t=3 & 	3.2775 & 	1.1516 & 	9.3265 & 	0.2864 & 	189.388\\
            t=4 & 	3.2778 & 	1.1515 & 	9.3288 & 	0.3251 & 	189.8009\\
            t=5 & 	3.2781 & 	1.1515 & 	9.3309 & 	0.3551 & 	190.0937\\
            t=6 & 	3.2785 & 	1.1515 & 	9.3327 & 	0.3779 & 	190.2949\\
            t=7 & 	3.2787 & 	1.1515 & 	9.3339 & 	0.3955 & 	190.4307\\
            t=8 & 	3.2786 & 	1.1514 & 	9.3341 & 	0.4089 & 	190.5091\\
            t=9 & 	3.2777 & 	1.151 & 	9.3317 & 	0.4193 & 	190.5048\\
            t=10 & 	3.274 & 	1.1496 & 	9.3206 & 	0.3307 & 	189.337\\
		\bottomrule
	\end{tabular}
\end{table}

\begin{figure}[htbp]
	\centering
        
	\includegraphics[width=0.45\textwidth]{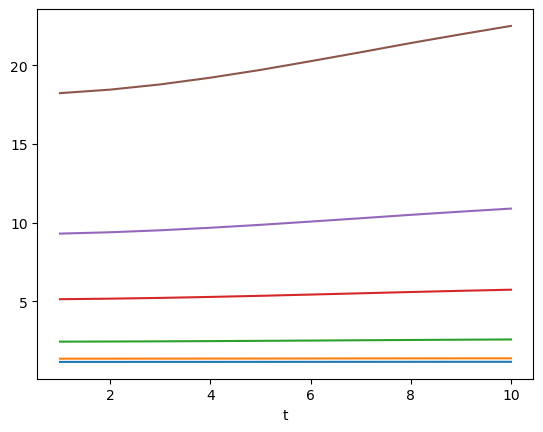}
        \includegraphics[width=0.45\textwidth]{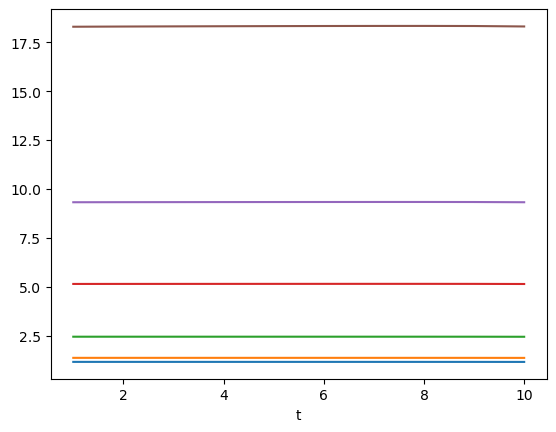}
         \caption{
         {\normalfont \small Case one, left graph, and case two, right graph. $MQ^{\rho}$. From the bottom curve to the top curve, $\rho$ fixed at $0.1, 0.2, 0.5, 0.8, 1, 1.2$ respectively. }}
         \label{figure:imperfect_upper}
\end{figure}

\section{Monitoring policies: transaction fees and clearing time randomization}\label{sec;bilevel}

\subsection{Bilevel optimization between the exchange and the strategic trader}

%A strategic trader benefits from arriving as late as possible while the regulator would prefer an early arrival in the market to restrict the price impact of a strategic trader.

The previous section emphasized some flaws in auctions and a need to regulate the arrival behavior of traders in this type of market. To mitigate the conflict between the optimal arrival $\hat \tau$ from the strategic trader viewpoint and $\hat \tau_{reg}$, the one that the exchange would prefer for market quality reasons, we investigate and provide a quantitative analysis of two tools: a transition fees policy and a randomization of the clearing time. We denote by $\xi(t)$ the transaction fee function, where $t$ refers to the arrival time of traders and $\tau^{cl}$ to be the duration of an auction assumed to be a random variable whose parameter is controlled by stock exchanges.\vspace{0.5em}

In this model, traders have to pay $\xi(t)$ per volume of shares of the asset traded if they arrive at time $t$ in the auction. We further assume that a transaction fee $\xi$ would impact the arrival intensity of market participants who will arrive less often when fees are larger. To be specific, under a fee structure $\xi$, the intensity of $N_t$ becomes $\lambda^{mm}_t = e^{- \xi(t)}$. Note that this choice of arrival intensity has been studied in \cite{AS08,EEMRT}. Let $\mathbb{P}^{\xi}$ be the measure under which $N_t$ has intensity $\lambda^{mm}_t = e^{- \xi(t)}$ and denote $\mathbb{E}^{\xi}$ to be the expectation with regard to the measure $\mathbb{P}^{\xi}$. \vspace{0.5em}

We then define for any $(t,n,(p_i)_{i=1}^n)$, the set $\hat M_{t,n,(p_i)_{i=1}^n}^{\xi,\tau^{cl}}$ of optimizers 
\begin{equation}
    \hat{\mu}(t,n,(p_i)_{i=1}^n;\xi,\tau^{cl}) \in \arg\max_{\mu} \mathbb{E}^{\xi}_{g}[\mathbf 1_{p_{\mu} \leq P_{\tau^{cl}}^{cl}} K(P^{cl}_{\tau^{cl}} - p_{\mu})\big(P^{cl}_{\tau^{cl}} - P^{*}-\xi(t)\big) | N_{t} = n, (P_i)_{i=1}^{N_t} = (p_i)_{i=1}^n ],
    \label{eq:choosemu}
\end{equation}
where $p_{\mu}$ refers to a normal random variable $\mathcal{N}(\mu,\sigma^2)$. The problem of the strategic seller becomes
% allowing a flexible auction duration
\begin{equation}\label{Vfee}
    V^{fee} (\xi,\tau^{cl}): = \sup_{\tau} V^{fee}_{\xi,\tau^{cl}}(\tau)
    \end{equation}
    with
    \[V^{fee}_{\xi,\tau^{cl}}(\tau)= \mathbb{E}^{\xi}_{g}\Big[\mathbf 1_{\tau \leq \tau^{cl}} \mathbf 1_{\hat P_{\tau} \leq P_{\tau^{cl}}^{cl}} \big\{ K(P^{cl}_{\tau^{cl}} - \hat P_{\tau})(P^{cl}_{\tau^{cl}} - P^{*}) - K(P^{cl}_{\tau^{cl}} - \hat P_{\tau})\xi(\tau) \big\}\Big],\]

where\footnote{Note that if the set $\hat M_{t,n,(p_i)_{i=1}^n}^{\xi,\tau^{cl}}$ is not reduced to one element, the choice of the optimizer does not affect the value function $V^{fee}$.} 
\begin{align*}
    &\hat{\mu}_{\tau} := \hat{\mu}(\tau,N_\tau,(P_i)_{i=1}^{N_\tau},\xi,\tau^{cl}),\text{ and }\hat{P}_{\tau} \sim \mathcal{N}(\hat{\mu}_{\tau},\sigma^2).
\end{align*}

We denote by $\widehat{\mathcal M}^{\xi,\tau^{cl}}$ the set of optimizers $(\hat\mu,\hat\tau)$ such that $\hat\tau$ is optimal for \eqref{Vfee} and $\hat\mu\in \hat M_{\hat\tau,N_{\hat\tau},(P_i)_{i=1}^{N_{\hat\tau}}}^{\xi,\tau^{cl}} $.\vspace{0.5em}

Assume that the strategic seller arrives in the auction at time $\hat\tau$ by proposing a price $\hat P_{\hat\tau}$. We denote by $P^{cl,\hat{\tau}}_{\tau^{cl}}$ the clearing price set by the exchange and defined by \eqref{eq:clear} with $P=\hat P_{\hat\tau}$ and $T=\tau^{cl}$. The problem of exchange depends on where its top interest lies. If the exchange wants to improve the actual price paid by the traders, the exchange aims at solving a bilevel optimization
%$\Lambda_0$ means that it is the value of the exchange computed from time 0

\begin{equation}\label{opt:neutral_mq}
   \Lambda_0 : = \min_{\xi} \Lambda_0(\xi) \end{equation}
   where
   \[\Lambda_0(\xi) = \min_{\tau^{cl},(\hat \tau, \hat \mu)}\mathbb{E}^{\xi}[ (|P^{cl,\hat{\tau}}_{\tau^{cl}} - P^{*}| + \frac{\sum_{i=1}^{N_{\tau^{cl}}} \xi(\tau_i)}{N_{\tau^{cl}}} )^2], \]
   or by assuming that the exchange is risk averse
    \begin{equation}\label{opt:rho_mq}\Lambda^{\rho}_0: = \min_{\xi} \Lambda^{\rho}_0(\xi)\end{equation}
    where
    \[ \Lambda^{\rho}_0(\xi)= \min_{\tau^{cl},(\hat \tau, \hat \mu)} \mathbb{E}^{\xi}\Big[\exp{ \Big(\rho\big(|P^{cl,\hat{\tau}}_{\tau^{cl}} - P^{*}| + \frac{\sum_{i=1}^{N_{\tau^{cl}}} \xi(\tau_i)}{N_{\tau^{cl}}} \big)\Big)}\Big], \]
    
%\begin{equation}\label{opt:neutral_mq}
%   \Lambda_0 : = \min_{\xi} \Lambda_0(\xi) \end{equation}
%   where
%   \[\Lambda_0(\xi) = \min_{\tau^{cl},(\hat \tau, \hat \mu)}\mathbb{E}^{\xi}[ (|P^{cl,\hat{\tau}}_{\tau^{cl}} - P^{*}| + \xi(\hat{\tau}) \mathbf 1_{ \hat{P}_{\hat\tau}\leq P^{cl,\hat{\tau}}_{\tau^{cl}}} )^2], \]
%   or   where    \[ \Lambda^{\rho}_0(\xi)= \min_{\tau^{cl},(\hat \tau, \hat \mu)} \mathbb{E}^{\xi}\Big[\exp{ \Big(\rho\big(|P^{cl,\hat{\tau}}_{\tau^{cl}} - P^{*}| + \xi(\hat{\tau}) \mathbf 1_{ \hat{P}_{\hat\tau}\leq P^{cl,\hat{\tau}}_{\tau^{cl}}} \big)\Big)}\Big], \]

 subject to  
 \begin{align*}
\text{(IC): }& (\hat \tau, \hat \mu)\in \widehat{\mathcal M}^{\xi,\tau^{cl}}, \\
\text{(R): }& V^{fee}(\xi,\tau^{cl}) \geq \frac{\gamma}2\mathbb{E}^{\xi}_{g}[ \mathbf 1_{\hat{P}_{\hat\tau}\leq P^{cl,\hat{\tau}}_{\tau^{cl}}} K(P^{cl,\hat{\tau}}_{\tau^{cl}} - \hat{P}_{\hat\tau})  ], 
\end{align*}

where $\frac{\sum_{i=1}^{N_{\tau^{cl}}} \xi(\tau_i)}{N_{\tau^{cl}}} $ is the average fee paid by market participant under a fee structure $\xi$ and $\gamma$ is the difference between the best bid price and the best ask price.   The constraint (IC) is called the incentive compatibility constraint and models the best-reaction action $(\hat\tau,\hat\mu)$ is the strategic seller when the exchange announced a transaction fee $\xi$ and a clearing time rule $\tau^{cl}$. The constraint (R) is set to bound the transaction fee $\xi$ and ensure that the trader would benefit from the auction and not turn to the continuous trading market to avoid the transaction fee in the periodic auction market. With this reservation utility constraint, the auction is more competitive than trading on the CLOB directly.\vspace{0.5em}

\begin{remark}
   This problem can by also seen as a ``trader focused'' exchange. We assume the exchange wants to minimize the total spread for the trader, where  
\begin{align*}
\text{total spread} = &MQ + \text{transaction spread} \\
= &\text{ spread between efficient price and clearing price  } +\\
 &\text{ spread between clearing price and after-fee price (real executed price)}. 
\end{align*}
    
    We define the transaction spread of a trader in this way: suppose the market clears at $P_T^{cl}$ and the transaction fee is $\xi$ per share, a buyer would pay $P_T^{cl} + \xi$ to buy a share and a seller would receive $P_T^{cl} - \xi$ to sell a share. Average transaction spread is thus $\frac{1}{2}[(P_T^{cl} + \xi) -(P_T^{cl} - \xi)] = \xi$.
\end{remark}
\vspace{0.5em}

If the exchange focuses only on the market efficiency and wants to increase both fee gains and market quality, the problem becomes

%$\Lambda_0$ means that it is the value of the exchange computed from time 0

\begin{equation}\label{opt:neutral_fee}
   \Lambda_0 : = \min_{\xi} \Lambda_0(\xi) \end{equation}
   where
   \[\Lambda_0(\xi) = \min_{\tau^{cl},(\hat \tau, \hat \mu)}\mathbb{E}^{\xi}[ |P^{cl,\hat{\tau}}_{\tau^{cl}} - P^{*}|^2 - \frac{\sum_{i=1}^{N_{\tau^{cl}}} \xi(\tau_i)}{N_{\tau^{cl}}} ], \]
   or by assuming that the exchange is risk averse
    \begin{equation}\label{opt:rho_fee}\Lambda^{\rho}_0: = \min_{\xi} \Lambda^{\rho}_0(\xi)\end{equation}
    where
    \[ \Lambda^{\rho}_0(\xi)= \min_{\tau^{cl},(\hat \tau, \hat \mu)} \mathbb{E}^{\xi}\Big[\exp{ \Big(\rho\big(|P^{cl,\hat{\tau}}_{\tau^{cl}} - P^{*}| - \frac{\sum_{i=1}^{N_{\tau^{cl}}} \xi(\tau_i)}{N_{\tau^{cl}}} \big)\Big)}\Big], \]

   %where \[\Lambda_0(\xi) = \min_{\tau^{cl},(\hat \tau, \hat \mu)}\mathbb{E}^{\xi}[ |P^{cl,\hat{\tau}}_{\tau^{cl}} - P^{*}|^2 - \xi(\hat\tau) \mathbf 1_{ \hat{P}_{\hat\tau}\leq P^{cl,\hat{\tau}}_{\tau^{cl}}} ], \]
   %or where \[ \Lambda^{\rho}_0(\xi)= \min_{\tau^{cl},(\hat \tau, \hat \mu)} \mathbb{E}^{\xi}\Big[\exp{ \Big(\rho\big(|P^{cl,\hat{\tau}}_{\tau^{cl}} - P^{*}| - \xi(\hat\tau) \mathbf 1_{ \hat{P}_{\hat\tau}\leq P^{cl,\hat{\tau}}_{\tau^{cl}}} \big)\Big)}\Big], \]

 subject to  
 \begin{align*}
\text{(IC): }& (\hat \tau, \hat \mu)\in \widehat{\mathcal M}^{\xi,\tau^{cl}}, \\
\text{(R): }& V^{fee}(\xi,\tau^{cl}) \geq \frac{\gamma}2\mathbb{E}^{\xi}_{g}[ \mathbf 1_{\hat{P}_{\hat\tau}\leq P^{cl,\hat{\tau}}_{\tau^{cl}}} K(P^{cl,\hat{\tau}}_{\tau^{cl}} - \hat{P}_{\hat\tau})  ].
\end{align*}

For the numerical solutions, we set the bid-ask spread $\gamma$ by referring to the estimation method in \cite{Farshid}: for a period of N days, $\gamma = \frac{1}{N}\sum_{i=1}^{N}\gamma_i$, where $$\gamma_i = \sqrt{\max\{4(c_t- (l_t+h_t)/2)(c_t- (l_{t+1}+h_{t+1})/2),0\}},$$ $c_t$ is the daily close log price, $l_t$ is the daily low log price, and $h_t$ is the daily high log price. By this method, we set Apple's $\gamma = 0.0039$ and Alphabet's $\gamma = 0.0065$. 

\begin{remark}
    Note that for both market impact or market efficiency optimization problems, we focus solely on the spread and fee of the strategic trader. We do not include the spread and fee of the other transferred limit orders because we assume these traders are transferred from CLOB to the periodic auction for execution and thus do not face the transaction fee imposed in the periodic auction.
\end{remark}

\subsection{Randomization without fees}

In this section, we focus on the solution to \eqref{opt:neutral_mq} or \eqref{opt:rho_mq}, \eqref{opt:neutral_fee} or \eqref{opt:rho_fee} when $\xi=0$. Echoing the discussion in Section \ref{sec;inefficiency}, we focus on case $(-)$, that is a misconception $\mu^*_{g} = \mu^* - \sigma$ and $\mu^{mm}_{g} = \mu^{mm} - \sigma$ for the strategic trader. We assume that $\tau^{cl}$ is a Bernoulli random variable taking values in $\{9,10\}$ with $p = \mathbb{P}(\tau^{cl}=9) = 1-\mathbb{P}(\tau^{cl}=10)$ and $p \in [0,1]$. The optimization on $\tau^{cl}$ in both \eqref{opt:neutral_mq} or \eqref{opt:rho_mq}, \eqref{opt:neutral_fee} or \eqref{opt:rho_fee} is reduced to optimize $p$, that is

\begin{equation*}
   MQ : = \min_{p\in[0,1]}   MQ (p) ,
   \end{equation*}

   % or by assuming that the exchange is risk averse
   %  \begin{equation*}MQ^\rho: &= {\min_{p\in[0,1]}   MQ^\rho (p)
   %  \end{equation*}

   with 
        \[ MQ (p)= \min_{(\hat \tau, \hat \mu)}\mathbb{E}[ |P^{cl,\hat{\tau}}_{\tau^{cl}} - P^{*}|^2  ], \]
   %     or
   %  \[MQ^\rho (p)= \min_{(\hat \tau, \hat \mu)} \mathbb{E}\Big[\exp{ \Big(\rho\big(|P^{cl,\hat{\tau}}_{\tau^{cl}} - P_T^{*}| \big)\Big)}\Big], \end{equation}
 subject to 
  \begin{align*}
\text{(IC): }& (\hat \tau, \hat \mu)\in \widehat{\mathcal M}^{0,\tau^{cl}}, \\
\text{(R): }& V^{fee}(0,p) \geq \frac{\gamma}2\mathbb{E}[ \mathbf 1_{\hat{P}_{\hat\tau}\leq P^{cl,\hat{\tau}}_{\tau^{cl}}} K(P^{cl,\hat{\tau}}_{\tau^{cl}} - \hat{P}_{\hat\tau})  ], 
\end{align*}

We recall that the strategic trader decides at time 0 when to arrive (i.e., chooses $\hat\tau$ before the auction starts). For both Apple and Alphabet, we observe in Table \ref{table:randnofee} that the  optimal $\hat p$ is around  $0.08$. \vspace{0.5em}

With randomization set at $\hat p = 0.08$, the stock exchange successfully encourages a strategic trader to arrive earlier. We now have $\hat{\tau}=9$ instead of $\hat{\tau}=10$ which is the optimal arrival time without randomization. We also observe the market quality improves from $3.6659$ to $3.6365$ for Apple and from $5.2687$ to $5.2263$ for Alphabet comparing with Table \ref{table:zerofee_imperfect}. Looking closely at the results in Table \ref{table:randnofee}, we observe that even a small $p$, as small as $0.08$ (which means that the auction only has $0.08$ probability to end at $9$ instead of $10$), would be sufficient to encourage the strategic trader to not arrive at $\tau=10$. This explains why the optimal $\hat p$ is very close to $0$: the exchange prefers the case when the strategic trader arrives before the closing time and at $p =0.08$ there is only $8\%$ of chance that the strategic trader arrives at the closing time of the auction which happens to be $9$.

\begin{table}[htbp]
	\caption{}
	\label{table:randnofee}
 \begin{center}
	\begin{tabular}{@{}c|cc|cc@{}}
		\toprule
            \textbf{Randomization $p$}& \multicolumn{2}{c|}{\textbf{Apple}} & \multicolumn{2}{c}{\textbf{Alphabet}} \\
		{}&  $MQ$ & $\hat{\tau}$ & $MQ$ &$\hat{\tau}$ \\
	       \midrule
		  0.0 & 	3.6659 & 	10 & 	5.2687 & 	10\\
            0.06 & 	3.6456 & 	10 & 	 5.2395 & 	10\\
            0.07 & 	3.6422 & 	10 & 	 5.2346 & 	10\\
            0.08 & 	3.6365 & 	9 & 	5.2263 & 	9\\
            0.09 & 	3.6374 & 	9 & 	5.2277 & 	9\\
            0.1 & 	3.6384 & 	9 & 	5.2291 & 	9\\
            0.5 & 	3.6764 & 	9 & 	5.2851 & 	9\\
            1.0 & 	3.7244 & 	9 & 	5.3526 & 	9\\
	       \bottomrule
	\end{tabular}
 \end{center}
\end{table}

\subsection{Optimal transaction fees indexed on time to improve price impact for the trader}\label{sec;mq1}

We now turn to the solutions of \eqref{opt:neutral_mq} and \eqref{opt:rho_mq}. As a classical result of contract theory and bilevel optimization, the shape of the contract $\xi$ has to be specified in a discrete time framework. We assume that the exchange proposes two type of exchange fees: either a linear fee indexed on the time of arrival of the strategic trader and the efficient price  $\xi_\ell(t) = at $ or a square fee structure $\xi_s(t) = at^2$, with randomization of closing time $\tau^{cl}$. By selecting either a linear fee structure or a square fee, the bilevel optimization problem thus becomes

\begin{equation}\label{opt:neutralmq}
   \Lambda_0 : = \min_{\xi\in\{\xi_\ell,\xi_s\}} \Lambda_0(a),\end{equation}
   where for a choice of $\xi\in\{\xi_\ell,\xi_s\}$
   \[ \Lambda_0(a)=\min_{\tau^{cl},(\hat \tau, \hat \mu)}\mathbb{E}^{\xi}[ (|P^{cl,\hat{\tau}}_{\tau^{cl}} - P^{*}| +\frac{\sum_{i=1}^{N_{\tau^{cl}}} \xi(\tau_i)}{N_{\tau^{cl}}}  )^2], \]
   
   or by assuming that the exchange is risk averse
    \begin{equation}\label{opt:rhomq}\Lambda^{\rho}_0: = \min_{a} \Lambda_0^{\rho}(a)
    \end{equation}
    where
    \[ \Lambda_0^\rho(a)=\min_{\tau^{cl},(\hat \tau, \hat \mu)} \mathbb{E}^{\xi}\Big[\exp{ \Big(\rho\big(|P^{cl,\hat{\tau}}_{\tau^{cl}} - P_T^{*}| + \frac{\sum_{i=1}^{N_{\tau^{cl}}} \xi(\tau_i)}{N_{\tau^{cl}}}  \big)\Big)}\Big], \]

 subject to  
 \begin{align*}
\text{(IC): }& (\hat \tau, \hat \mu)\in \widehat{\mathcal M}^{\xi,\tau^{cl}}, \\
\text{(R): }& V^{fee}(\xi) \geq \frac{\gamma}2\mathbb{E}^{\xi}_{g}[ \mathbf 1_{\hat{P}_{\hat\tau}\leq P^{cl,\hat{\tau}}_{\tau^{cl}}} K(P^{cl,\hat{\tau}}_{\tau^{cl}} - \hat{P}_{\hat\tau})  ].
\end{align*}

As before, we assume that $p = \mathbb P(\tau^{cl}=9)=1-\mathbb P(\tau^{cl}=10)$ and the exchange optimizes on the parameter $p$ to optimize $\tau^{cl}$. In addition, as the numerical method discretizes the time span $[0,T]$, the average fee $\frac{\sum_{i=1}^{N_{\tau^{cl}}} \xi(\tau_i)}{N_{\tau^{cl}}}$ computed by the numerical method is actually $\frac{\sum_{t=1}^{\tau^{cl}} \xi(t)(N_{t} - N_{t-1})}{N_{\tau^{cl}}}$ as we assume traders who arrive between time $(t-1)$ and time $t$ pay $\xi(t)$ for $t \in {1,...,T}$.  \vspace{0.5em}

The results are presented in Table \ref{table:section33} with the optimal fee structure for Apple ('Stock'-'Apple' cells) and Alphabet ('Stock'-'Alphabet' cells). The second column ``MQ measure'' represents the choice of the problem  \eqref{opt:neutralmq} or \eqref{opt:rhomq} for different values of $\rho$. The third column gives the optimal fees' structure while the fourth column gives the optimal value of $p$. The fifth column gives the optimal arrival of the strategic trader $\hat \tau$ solving \eqref{Vfee} with the corresponding optimal fees and $p$ value. The sixth column ``Exchange gain'' represents the value $\Lambda_0$ or $\Lambda_0^{\rho}$ for different values of $\rho$ corresponding to the value function of the exchange, while the last column compared this value with the market quality $MQ(\hat\tau)$ when the fee is $0$ (no transaction fees).\vspace{0.5em}

From the last two columns of the table, we could see that our transaction fee model improves market quality since $\Lambda_0$ is always smaller that the market quality when $\xi=0$. From the strategic trader's perspective, they now have the incentive to arrive earlier as their optimal arriving time $\hat \tau$ is less than $10$ for several of the cases shown here. We could thus conclude that the added fee is effective to cure the flaws of a periodic auction system observed in Section \ref{sec;inefficiency}. \vspace{0.5em}

\begin{table}[htbp]
    \centering
    \caption{Selected Numerical Results}
    \label{table:section33}
    \begin{tabular}{@{} c|c| c| c| c| c |c @{}} 
        \toprule
        \addlinespace[1ex]
        %Stock & MQ Measure & Optimal Fee and Randomization &  $\hat{\tau}$ & Exchange's Gain & Fee Gain & MQ & MQ when no fee \\
        Stock & MQ Measure & Optimal Fee & Randomization&  $\hat{\tau}$ & Exchange's Value & $MQ^0$\\
        %\cline{1-8} 
        \midrule
        Apple & $\Lambda_0$ & $0.003t^2$ & $p=0$ & 6 & 3.553 & 3.666 \\[0.3em]
        & $\Lambda_0^{0.01}$ &  $0.001t^2$& $p=0$ & 10 & 0.999 & 1.004 \\[0.3em]
        & $\Lambda_0^{0.5}$ & $0.001t^2$& $p=0$ & 10 & 2.598 & 2.599 \\[0.3em]
        & $\Lambda_0^{1.5}$ &$0.002t^2$ & $p=0$ & 7 & 76.001 & 80.180 \\
        \midrule
        Alphabet & $\Lambda_0$ & $0.004t^2$ & $p=0$ & 5 & 5.064 & 5.269 \\[0.3em]
        & $\Lambda_0^{0.01}$ &  $0.001t^2$& $p=0$ & 10 & 1.002  & 1.007 \\[0.3em]
        & $\Lambda_0^{0.5}$ & $0.001t^2$& $p=0.1$ & 9 & 3.313 & 3.326 \\[0.3em]
        & $\Lambda_0^{1.5}$ &$0.003t^2$ & $p=0$ & 6 & 294.589 & 318.554 \\
        \bottomrule
    \end{tabular}
\end{table}

We now turn to a deeper study of how our transaction fee model improves the market's quality for the Apple's stock price.\footnote{Alphabet's graphs carry the same spirit as Apple's graphs. For the sake of simplicity, we only analyze the data from Apple in this part.} When the fees increase initially, the spread term $|P_T^{cl} - P^*|$ becomes better due to the fact that the strategic trader is willing to arrive earlier. Note that both the linear transaction fee structure and square transaction fee structure encourage the strategic seller to arrive earlier in the market. From Figure \ref{fig:result33} (a) and (b), we see that as the fee increases (led by increasing $a$), the strategic seller's optimal $\tau$ gradually declines from $10$ to $1$ and the decline rate seems to coincide with the structure of the fee model as Figure \ref{fig:result33} (a) shows a linear declining pattern and Figure \ref{fig:result33} (b) shows an accelerating declining pattern.\vspace{0.5em}

However, if the exchange continues to increase the transaction fee, the increasing fee term $\xi$ in the stock exchange's utility function $\Lambda_0$ and $\Lambda_0^{\rho}$ starts to overpower the benefit. In addition, a larger fee would discourage market participants to join the market (as intensity $\lambda^{mm}$ decreases) and according to \ref{remark:moretrade}, market quality deteriorates when there are less traders in the market. \vspace{0.5em}

\begin{figure}[htbp]
    \centering
    \begin{minipage}{0.45\textwidth}
        \centering
        \includegraphics[width=\linewidth]{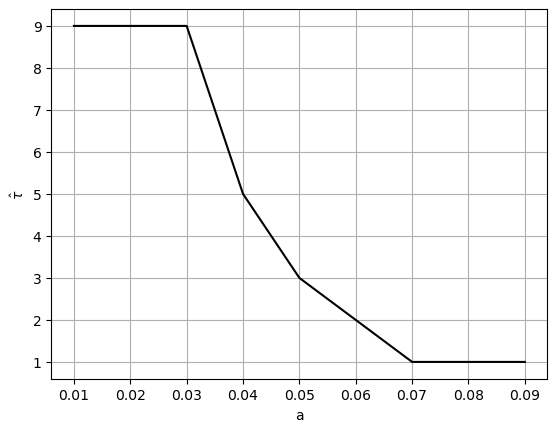}
        \caption*{(a) Linear Fee $\hat{\tau}(a)$}
    \end{minipage}
    \begin{minipage}{0.45\textwidth}
        \centering
        \includegraphics[width=\linewidth]{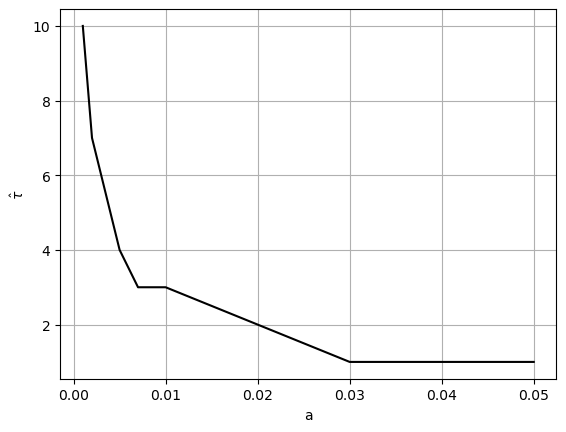}
        \caption*{(b) Square Fee $\hat{\tau}(a)$}
    \end{minipage}
    \begin{minipage}{0.45\textwidth}
        \centering
        \includegraphics[width=\linewidth]{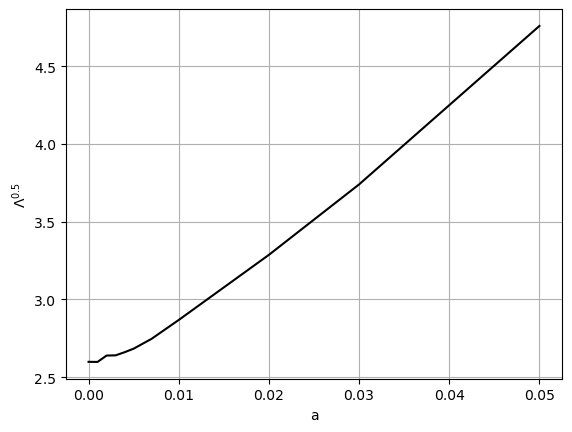}
        \caption*{(c) Under square Fee $\Lambda_0^{0.5}(a)$ }
    \end{minipage}
    \begin{minipage}{0.45\textwidth}
        \centering
        \includegraphics[width=\linewidth]{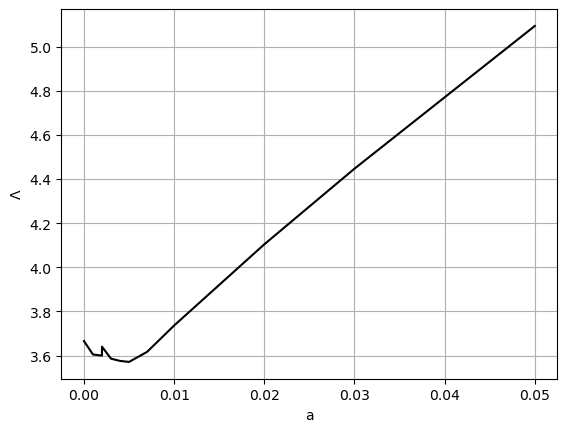}
        \caption*{(d) Under square Fee $\Lambda_0(a)$ }
    \end{minipage}
    \caption{Numerical analysis of the optimal arrival of the strategic seller (a) and (b), the value of the exchange with risk aversion parameter $\rho =0.5$ (c), the value of the exchange without risk aversion parameter (d). Note that $\hat{\tau}$ is a function of both $a$ and $p$; for example, under linear fee, $\hat{\tau}(a=0.06,p=0) =1$ and $\hat{\tau}(0.06, 0.2) =2$. (a) and (b) show $\hat{\tau}(a,p(a))$, where $p(a) = \arg\min_{p} \Lambda_0^{0.5}(a,p)$. }

    \label{fig:result33}
\end{figure}

Note moreover that $p=0$ is optimal for all except one cases shown in table \ref{table:section33}, that is the randomization of the clearing time becomes useless by adding transaction fees. This is likely due to the fact that fees discourage the number of traders in the market by decreasing $\lambda^{mm}$. However, to achieve better market quality, the exchange would prefer more traders to join, so would prefer to hold the auction for a longer time span ($p=0$ means the market always close at time $t=10$ instead of $t=9$).\vspace{0.5em}

\subsection{Optimal transaction fees indexed on time: improving market quality while benefiting from the fees}\label{sec:qualityfees}

We now turn to the solutions of \eqref{opt:neutral_fee} and \eqref{opt:rho_fee}. The bilevel optimization problem thus becomes

\begin{equation}\label{opt:neutralfee}
   \Lambda_0 : = \min_{\xi\in\{\xi_\ell,\xi_s\}} \Lambda_0(a),\end{equation}
   where for a choice of $\xi\in\{\xi_\ell,\xi_s\}$
   \[ \Lambda_0(a)=\min_{\tau^{cl},(\hat \tau, \hat \mu)}\mathbb{E}^{\xi}[ |P^{cl,\hat{\tau}}_{\tau^{cl}} - P_T^{*}|^2 - \frac{\sum_{i=1}^{N_{\tau^{cl}}} \xi(\tau_i)}{N_{\tau^{cl}}}  ], \]
   or by assuming that the exchange is risk averse
    \begin{equation}\label{opt:rhofee}\Lambda^{\rho}_0: = \min_{a} \Lambda_0^{\rho}(a)
    \end{equation}
    where
    \[ \Lambda_0^\rho(a)=\min_{\tau^{cl},(\hat \tau, \hat \mu)} \mathbb{E}^{\xi}\Big[\exp{ \Big(\rho\big(|P^{cl,\hat{\tau}}_{\tau^{cl}} - P_T^{*}| - \frac{\sum_{i=1}^{N_{\tau^{cl}}} \xi(\tau_i)}{N_{\tau^{cl}}}  \big)\Big)}\Big], \]
 subject to  
 \begin{align*}
\text{(IC): }& (\hat \tau, \hat \mu)\in \widehat{\mathcal M}^{\xi,\tau^{cl}}, \\
\text{(R): }& V^{fee}(\xi) \geq \frac{\gamma}2\mathbb{E}^{\xi}_{g}[ \mathbf 1_{\hat{P}_{\hat\tau}\leq P^{cl,\hat{\tau}}_{\tau^{cl}}} K(P^{cl,\hat{\tau}}_{\tau^{cl}} - \hat{P}_{\hat\tau})  ].
\end{align*}

As before, we assume that $p = \mathbb P(\tau^{cl}=9)=1-\mathbb P(\tau^{cl}=10)$ and the exchange optimizes on the parameter $p$ to optimize $\tau^{cl}$.

\begin{remark}\label{rem:auctioncost}
    In an informal way, we can see from \eqref{opt:neutralfee} or \eqref{opt:rhofee} that
    \[ \text{Exchange value function} = \text{market quality cost } - \text{ fees},\]
or in other words
\[\text{Market quality cost} = \text{Exchange gain } + \text{ Fees' gain}.\]\end{remark}

The results are presented in Table \ref{table:conclusion34}. The seventh column is the incomes generating by the transaction fees, that is the part $\xi(\hat\tau)$ calculated in an informal way with column "$MQ^{\xi}$" - "Exchange's Gain". The eigth column shows the market quality $MQ(\hat\tau)$ by considering the optimal fees $\xi$ while the last column compared this value with the market quality $MQ(\hat\tau)$ when the fee is $0$ (no transaction fees).\vspace{0.5em}

%A first remark from the second and third columns is that the linear fee structure $\xi_\ell$ is optimal for the risk-neutral exchange problem \eqref{opt:neutral_fee} while the square fee structure $\xi_s$ is optimal for a risk-averse exchange solving \eqref{opt:rho_fee}.\vspace{0.5em}

\begin{table}[htbp]
    \centering
    \caption{Selected Numerical Results}
    \label{table:conclusion34}
    \begin{tabular}{@{} c|p{1.5cm}| c| p{1.6cm}| c| *{3}{c|}c @{}} 
        \toprule
        \addlinespace[1ex]
        %Stock & MQ Measure & Optimal Fee and Randomization &  $\hat{\tau}$ & Exchange's Gain & Fee Gain & MQ & MQ when no fee \\
        Stock & MQ Measure & Optimal Fee & Random-ization&  $\hat{\tau}$ & Exchange's Gain & Fee Gain & $MQ^{\xi}$ & $MQ^0$ \\
        \cline{1-8} 
        \midrule
        Apple & $\Lambda_0$ & $0.24t^2$ & $p=0$ & 1 &  0.417 & 3.546 &  3.963 & 3.666 \\[0.3em]
        & $\Lambda_0^{0.01}$ &$0.24t^2$ & $p=0$ & 1 &  0.971 & 0.035 &  1.006& 1.004  \\[0.3em]
        & $\Lambda_0^{0.5}$ & $0.24t^2$ & $p=0$ & 1 &  0.466 & 2.276 &  2.742 & 2.599 \\[0.3em]
        & $\Lambda_0^{1.5}$ & $0.23t^2$ & $p=0$ & 1 & 0.768 & 153.904  & 154.672 & 80.180 \\
        \midrule
        Alphabet & $\Lambda_0$ & $0.23t^2$ & $p=0$ & 1 & 2.173  & 3.537  &  5.710 & 5.269 \\[0.3em]
        & $\Lambda_0^{0.01}$ &$0.24t^2$ & $p=0$ & 1 & 0.974  & 0.035 & 1.009 & 1.007 \\[0.3em]
        & $\Lambda_0^{0.5}$ & $0.24t^2$ & $p=0$ & 1 & 0.610  & 2.981 & 3.591  & 3.326 \\[0.3em]
        & $\Lambda_0^{1.5}$ & $0.22t^2$ & $p=0$ & 1 & 5.026 & 982.068 & 987.094 & 318.554 \\
        \bottomrule
    \end{tabular}
\end{table}

\begin{figure}[htbp]
    \centering
    \includegraphics[width=0.5\textwidth]{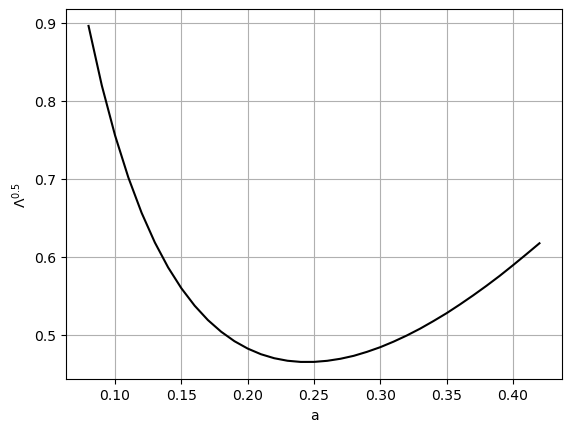}
    \caption{Graph of $a$ v.s. $\Lambda_0^{0.5}(a,p(a))$, which is the value of the exchange with risk aversion parameter $\rho =0.5$ under square fee structure.}
% The value of the exchange in the risk-neutral case, (d).
    \label{fig:result34}
\end{figure}

Given the exchange is both market quality and fee driven, we see that the optimal transaction fees are different from the previous section. The fee part becomes the main driver in approaching the minimum of $\lambda_0(a)$ or $\lambda_0^{\rho}(a)$. From figure \ref{fig:result34}, we see that as fee increases the exchange becomes better significantly due to larger fee gains. However, the concavity of the graph also implies a declining fee benefit as fee increases. On the one side, this is due to the competition between the fee term $\xi$ and the spread term $|P_T^{cl} - P^*|$ in the stock exchange's utility function $\Lambda_0$ and $\Lambda_0^{\rho}$. A larger fee discourages traders to join thus harms market quality. On the other side, a larger fee per share would lead to less traders in the market thus less fee payers in the market. Therefore, the graph finally reaches an optimal point (a minimum point) under these considerations. From the strategic trader's perspective, they now have the incentive to arrive at time $\tau = 1$ instead of $\tau = 10$ to avoid late arrival "penalties". \vspace{0.5em}

Note moreover that $p=0$ is optimal since a longer time span allows the auction to receive more traders thus bringing larger fee gains for the exchange.\vspace{0.5em}

%Therefore, the exchange finds a balance at $0.04t$ between better market quality and better fee gains. As a matter of fact, the regulator calibrated the fees optimally to make the strategic seller arriving in the auction as early as possible. It deletes the impact of the randomization of the clearing time between $\tau^{cl}=9$ or $\tau^{cl}=10$. This self-adjustment of the transaction fee acts like our CLOB constraint (R) that restricts the transaction fee to go too large. 

%which compensate for the inefficiency they bring to the system. However, when $a$ reaches a certain level ($0.08$ in our case), the gain of exchanges diminishes as traders now prefer to arrive earlier to avoid the high fees paid to the exchanges. 

\section{Conclusions}\label{sec;ccl}

We study the strategic arrival of a trader and show that a strategic trader always joins an auction at the last moment to have the greatest price manipulation power. Such behavior could impair fairness and market quality especially if the trader has a misconception of the efficient price. We propose two solutions: randomizing the closing time and introducing time-dependent transaction fees. With randomization ($92\%$ of chance auction closes at $T=10$ and $8\%$ of chance closes at $T=9$), the strategic trader would join the auction before the last moment and market quality is improved. With transaction fees, we consider two possible interests of the exchange, improving the market quality and making fee gains. Under either considerations, our solution provides better results than without randomization or fees. 

\bigskip

Our results certainly have limitations. We assume the presence of a single strategic trader instead of allowing  multiple strategic traders to compete. In terms of model setting, \cite{gayduk} studies a general optimal control and stopping problem with discrete controls and proves the existence of an equilibrium in a game in which every player is strategic, without studying randomization or transaction fees policy. Our paper is a specific optimal control and stopping problem with discrete stopping time and we assume all but one players are non-strategic. \cite{gayduk} sheds light on a possible extension of our model to a more general case when there are more than one strategic players in an auction. \cite{alfonsi} studies the optimal execution strategy of a strategic trader in a continuous market whose order has price impact on the market and finds that the existence of price manipulation strategies depend on the choice of models. This reminds us that changing certain settings of our model may change the behavior of the strategic trader and possibly our conclusion. Updating the model to include interactions of multiple strategic traders would probably lead to similar results of this paper. Strategic traders would very likely still choose to arrive at the last moment to avoid sharing their views too early and to gain information of others to take advantage. We also assume a discrete optimal stopping control instead of a continuous one. The continuous version of this problem is in the works. We assume the auction market clears all order imbalance while in reality order imbalance exists. Finding a workable model to allow order imbalance is a future direction to consider. In addition, we model simply one round of an auction. If running the auction for several rounds, we could possibly see a larger negative impact of the strategic trader's arrival-timing strategy and see a greater need to regulate the arrival of traders. 

\bigskip

Finally, we did not mention priority rules for a periodic auction. This is because our model assumes zero imbalance of orders, so we do not find the need of a priority rule. In reality, order imbalance exists; for example, if A wants to sell 10 shares of stock, B wants to buy 4 shares of stock, and C wants to buy 8 shares of stock, B or C or both might only receive part of what they request. Priority rules need to be set to divide the 10 shares in this situation. Cboe's periodic auction market assigns price priority over size priority over time priority; \cite{Budish} also mentions that price priority should be given over time priority. In general, time priority should be the last to consider. Therefore, we argue that our observation (a strategic trader lacks an incentive to join early) would still be valid given the presence of order imbalance and priority rules. However, learning how to set priority rules for order imbalance could be a meaningful future study.

\appendix
\section{Appendix: Numerical Methods}
\label{appendix:discrete}

\subsection{Problem of a Strategic Seller}

Recall the objective of a strategic seller is $V^\circ = \sup_{\tau} V^\circ(\tau)$. Fix $(t, n,\{p_i\}_{i=1}^{n})$, recall the price $P(t, n,\{p_i\}_{i=1}^{n})$ submitted by the strategic seller is a normal random variable.  \vspace{0.5em}

Denote $q(p,x, y)$ to be the joint density function of $P(t, n,\{p_i\}_{i=1}^{n})$, $P^{*}$,  $\sum\limits_{j: j =1, \tau_j \geq t}^{m } P_j$, where $\tau_j$ is the arrival time of the $j-th$ market maker with price $P_j$. Each of the three random variables follows a normal distribution. For simplicity, assume the three normal random variables are mutually independent. \vspace{0.5em}

%We test two cases, case one is when $P(t, n,\{p_i\}_{i=1}^{n})$, $P_T^{*}$, $\{P_j\}_{\tau_j \geq t}$ are independent normal random variables, case two is when these variables has correlation 0.6. The result does not change. 

Define $N_{(t,T)} :=  N_T - N_t$. Assume $N_{(t,T)}$ is independent of $P(t, n,\{p_i\}_{i=1}^{n})$ , $P^{*}$, $\{P_j\}_{\tau_j \geq t}$. Denote $f_{N_{(t,T)}}$ to be the probability density function of $N_{(t,T)}$. \vspace{0.5em}

Fix $N_t = n, \{P_i\}_{i=1}^{N_t} = \{p_i\}_{i=1}^n$, $\hat{\mu}(t, n,\{p_i\}_{i=1}^n)$ is the optimizer of :

\begin{align*}
    &\sup_{\mu}  \mathbb{E} \biggr[\mathbf 1_{P_{\mu} \leq  \frac{ \sum_{i=1}^{N_T} P_i}{N_T }  } \biggl\{ p_{\mu}^2 \frac{-K N_T}{(1+N_T)^2} + p_{\mu} \left[\frac{ K\sum_{i=1}^{N_T} P_i -  K N_T \sum_{i=1}^{N_T} P_i}{(1+N_T)^2} + \frac{P^{*}_T  K N_T }{(1 + N_T)} \right] + 	\frac{K (\sum_{i=1}^{N_T} P_i)^2}{(1 + N_T)^2}\\
    &\qquad - \frac{P^{*}K \sum_{i=1}^{N_T} P_i}{(1 + N_T)}\biggl\} \Big| N_{t} = n, (P_i)_{i=1}^{N_t} = (p_i)_{i=1}^{N_t}\biggr]\\
    = &\sup_{\mu} \scalebox{1}[1.5]{$\displaystyle\int$}  \mathbb{E}\biggr[\mathbf 1_{P_{\mu} \leq  \frac{ \sum_{i=1}^{N_T} P_i}{N_T }  } \biggl\{ p_{\mu}^2 \frac{-K N_T}{(1+N_T)^2} + p_{\mu} \left[\frac{ K\sum_{i=1}^{N_T} P_i -  K N_T \sum_{i=1}^{N_T} P_i}{(1+N_T)^2} + \frac{P^{*}_T  K N_T }{(1 + N_T)} \right] + \frac{K (\sum_{i=1}^{N_T} P_i)^2}{(1 + N_T)^2}\\
    &\qquad	- \frac{P^{*}K \sum_{i=1}^{N_T} P_i}{(1 + N_T)}		 \biggl\}\Big| N_{t} = n, (P_i)_{i=1}^{N_t} = (p_i)_{i=1}^{N_t}, N_{(t,T)} = m\biggr] dN_{(t,T)} (m)\\
    =& \sup_{\mu} \sum_{m=0}^{\infty} f_{N_{(t,T)}}(m)\biggl\{ \mathbb{E}\biggr[\mathbf 1_{P_{\mu} \leq  \frac{ \sum^{n} p_i+ \sum^{m } P_j}{n+m} }  \big\{-Kp_{\mu}^2\frac{n + m}{(n+ m +1)^2} +Kp_{\mu}\frac{	(\sum^{n} p_i + \sum^{m } P_j)(1- n- m )}{(n + m +1)^2} \\
    & \qquad  + K p_{\mu} P^{*}\frac{ n+ m}{(n +m +1)} +K\frac{(\sum^{n} p_i + \sum^{m} P_j)^2	}{(n + m +1)^2}  - KP_T^{*}\frac{	(\sum^{n} p_i + \sum^{m} P_j)	}{(n + m +1)} \big\}\biggr]  \biggl\} \\
    =& \sup_{\mu} \sum_{m=0}^{\infty} f_{N_{(t,T)}}(m)\biggl\{  \int_{-\infty}^{\infty}\int_{-\infty}^{\infty}\int_{-\infty}^{\infty} \mathbf 1_{P_{\mu} \leq  \frac{ \sum^{n} p_i+ y}{n+m} }  \big\{-Kp_{\mu}^2\frac{n + m}{(n+ m +1)^2} +Kp_{\mu}\frac{	(\sum^{n} p_i + y)(1- n- m )}{(n + m +1)^2}\\
    &\qquad  + K p_{\mu} x\frac{ n+ m}{(n +m +1)} +K\frac{(\sum^{n} p_i + y)^2	}{(n + m +1)^2}  - Kx\frac{	(\sum^{n} p_i + y)	}{(n + m +1)} \big\}q(p_{\mu},x,y) d(p_{\mu},x,y)  \biggl\}.
\end{align*}
where the second equality is due to independence.

\bigskip

\subsection{Appendix: Problem of the Regulator}

\begin{align*}
    \mathbb{E}\left[\left|P^{cl,t}_{T} - P^{*}\right|^2\right] &= \mathbb{E}\left[\bigg| \frac{\sum^{N_T} P_i +  1_{\left\{P_t \leq \frac{\sum^{N_T}P_i}{N_T} \right\}} P_t}{N_T +1_{ \left\{P_t \leq \frac{\sum^{N_T}P_i}{N_T}\right\} } }- P^{*}\bigg|^2\right]  \\
    &=\scalebox{1}[2]{$\displaystyle\int$} \mathbb{E}\left[\biggl( \frac{\sum^{N_T} P_i +  1_{\left\{P_t \leq \frac{\sum^{N_T}P_i}{N_T} \right\}} P_t}{N_T +1_{ \left\{P_t \leq \frac{\sum^{N_T}P_i}{N_T}\right\} } }- P^{*}\biggl)^2\Bigg|{N_{(t,T)}}\right] dN_{(t,T)}\\
    &=\scalebox{1}[2]{$\displaystyle\int$} \mathbb{E}\mathbb{E}\left[\biggl( \frac{\sum^{N_T} P_i +  1_{\left\{P_t \leq \frac{\sum^{N_T}P_i}{N_T} \right\}} P_t}{N_T +1_{ \left\{P_t \leq \frac{\sum^{N_T}P_i}{N_T}\right\} } }- P^{*}\biggl)^2\Bigg|\mathcal{F}_t, {N_{(t,T)}}\right] dN_{(t,T)}\\
\end{align*}

%We have $\mathbb{E}\left[\biggl( \frac{\sum^{N_T} P_i +  1_{\left\{P_t \leq \frac{\sum^{N_T}P_i}{N_T} \right\}} P_t}{N_T +1_{ \left\{P_t \leq \frac{\sum^{N_T}P_i}{N_T}\right\} } }- P^{*}\biggl)^2\Bigg|{N_{(t,T)}}\right] = \mathbb{E}\mathbb{E}^{\mathcal{F}_t}\left[| \frac{\sum^{N_T} P_i +  1_{\left\{P_t \leq \frac{\sum^{N_T}P_i}{N_T} \right\}} P_t}{N_T +1_{ \left\{P_t \leq \frac{\sum^{N_T}P_i}{N_T}\right\} } }- P^{*}|^2|{N_{(t,T)}}\right] $, where the first $\mathbb{E}$ is with regard to $N_t, \{P_i\}_{\tau_i <t }$, which is calculable by independence of $N_t$ and $\{P_i\}_{\tau_i <t }$.

By Disintegration Theorem \cite{Kall}, 

\begin{align*}
    \mathbb{E}\left[\biggl( \frac{\sum\limits^{N_T} P_i +  1_{\left\{P_t \leq \frac{\sum^{N_T}P_i}{N_T} \right\}} P_t}{N_T +1_{ \left\{P_t \leq \frac{\sum^{N_T}P_i}{N_T}\right\} } }- P^{*}\biggl)^2\Bigg|\mathcal{F}_t,N_{(t,T)}\right] = \scalebox{1}[2]{$\displaystyle\int$}  \Bigg| \frac{\sum\limits^{N_t} P_i +y + 1_{\left\{p_t \leq \frac{\sum^{N_t}P_i+y }{{N_{(t,T)}}+N_t} \right\}} p}{{N_{(t,T)}}+N_t +1_{ \left\{p \leq \frac{\sum^{N_t}P_i+y }{{N_{(t,T)}}+N_t}\right\} } }- x\Bigg|^2	q(p,x,y) d(p,x,y),
\end{align*}
where $p$ refers to $P$, $x$ refers to $P^{*}$, $y$ refers to $\sum\limits^{N_{(t,T)}} P_j$.

%\begin{align*}
%&\mathbb{E}^{\mathcal{F}_t}\left[| \frac{\sum^{N_T} P_i +  1_{\left\{P_t \leq \frac{\sum^{N_T}P_i}{N_T} \right\}} P_t}{N_T +1_{ \left\{P_t \leq \frac{\sum^{N_T}P_i}{N_T}\right\} } }- P^{*}|^2|{N_{(t,T)}}\right](\omega) \\
%&= \int | \frac{\sum^{N_t(\omega)} P_i(\omega) +y + 1_{\left\{p_t \leq \frac{\sum^{N_t(\omega)}P_i(\omega)+y }{{N_{(t,T)}}+N_t(\omega)} \right\}} p}{{N_{(t,T)}}+N_t(\omega) +1_{ \left\{p \leq \frac{\sum^{N_t(\omega)}P_i(\omega)+y }{{N_{(t,T)}}+N_t(\omega)}\right\} } }- x|^2	q(p,x,y) d(p,x,y),
%\end{align*}

\section{Appendix: Illustrate Remark \ref{remark:hatmu}}
\label{appendix:illustrate}

To better illustrate the proof, we draw a figure \ref{graph:proofillu} of $$y = -(KN_T)x + K\sum_{i=1}^{N_T}P_i - K N_TP^*.$$ For each point on the line, the $x$ coordinate would be $$x(p) = P_T^{cl}(p) - P^* = \frac{\sum^{N_T}P_i + p}{N_T +1} - P^*$$ and the $y$ coordinate would be $$y(p) = K(P_T^{cl}(p) - p) = K(\frac{\sum^{N_T}P_i + p}{N_T +1} - p),$$ where p is the price sent by the strategic seller. The shadow area is the trader's gain.

We have $\bar{\mu}(\omega)$ achieved at the middle point of the line segment and $P_{\bar{\mu}(\omega)}$ is normally distributed as $\mathcal{N}(\bar{\mu}(\omega),\sigma^2)$.

\begin{figure}[H]
    \centering
    \includegraphics[width=0.5\textwidth]{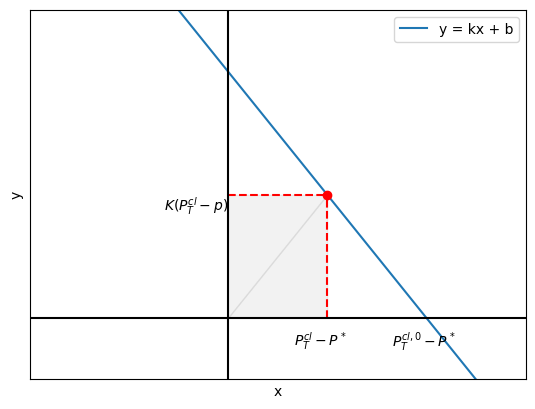}
    \caption{}
    \label{graph:proofillu}
\end{figure}

\bibliography{template} %-->reference list is on the template.bib file

\end{document}